%% file: main.tex
\documentclass[prd,twocolumn,aps,letterpaper,amsmath,amssymb,preprintnumbers,showpacs,superscriptaddress,floatfix,longbibliography,nofootinbib]{revtex4-2}
\usepackage{array}
\usepackage{graphicx}
\usepackage[export]{adjustbox}
\usepackage[caption=false]{subfig}
\usepackage{tikz}
\usetikzlibrary{tikzmark}
\usetikzlibrary{calc}
\allowdisplaybreaks %
\usepackage[hypertexnames=true]{hyperref}
\usepackage[capitalize]{cleveref}
\hypersetup{
    colorlinks=true,       %
    linkcolor=blue,          %
    citecolor=blue,        %
    filecolor=blue,      %
    urlcolor=blue           %
}

\usepackage{dsfont}
\usepackage{slashed}
\usepackage{mathtools}
\usepackage{listings}
\usepackage{pgf}
\usepackage{fancyvrb}
\usepackage{longtable}
\usepackage{bm}
\usepackage{tikz-cd}
\usepackage{amsthm}
\usepackage{mathrsfs}
\usepackage{algorithm}
\usepackage{algpseudocode}
\usepackage[normalem]{ulem} %
\usepackage{graphicx} %
\AddToHook{cmd/appendix/before}{%
    \crefalias{section}{appendix}%
    \crefalias{subsection}{appendix}
}

\newtheorem{lemma}{Lemma}

\DeclareMathOperator{\imag}{Im}
\DeclareMathOperator{\real}{Re}

\DeclareMathOperator{\diag}{diag}

\DeclareMathOperator{\Tr}{Tr}

\DeclareMathOperator{\cl}{cl}
\newcommand{\D}[0]{\ensuremath{\mathbb{D}}}

\newcommand{\R}[0]{\ensuremath{\mathbb{R}}}
\newcommand{\C}[0]{\ensuremath{\mathbb{C}}}

\newcommand{\symmat}[1]{\ensuremath{\mathbb{S}^{#1}}}
\newcommand{\hermmat}[1]{\ensuremath{\mathscr{H}^{#1}}}
\newcommand{\Ncorr}{\ensuremath{N_{\text{corr}}}}

\usepackage{braket}

\newcommand{\norm}[1]{||#1||}

\newcommand{\matdot}{\bullet}

\DeclarePairedDelimiter\floor{\lfloor}{\rfloor}

\newcommand{\minimize}{\operatorname*{minimize}}
\newcommand{\maximize}{\operatorname*{maximize}}
\newcommand{\st}{\operatorname*{subject~to}}

\newtheorem{optprob}{Problem}[section]

\newenvironment{optproblem}[1]%
{\begin{optprob}[#1]\begin{equation}\begin{aligned}}%
{\end{aligned}\end{equation}\end{optprob}}

\newcommand{\covscale}[0]{\alpha}

\date{\today}

\begin{document}
\title{The Causal Bootstrap:\\
Bounding Smeared Spectral Functions from Non-Perturbative Euclidean Data}
\begin{abstract}
\input{abstract.tex}
\end{abstract}

\author{Ryan Abbott}
\email{rwa2110@columbia.edu}
\affiliation{Physics Department, Columbia University, New York, NY 10027, USA}

\author{Sarah Fields}
\affiliation{Physics Department, Columbia University, New York, NY 10027, USA}

\author{William I. Jay}
\affiliation{Department of Physics, Colorado State University, Fort Collins, CO 80523, USA}

\author{Patrick Oare}
\affiliation{Physics Department, Brookhaven National Laboratory, Upton, NY 11973, USA}

\author{Matteo Saccardi}
\affiliation{Department of Physics, Colorado State University, Fort Collins, CO 80523, USA}

\maketitle

\input{sec1_intro.tex}
\input{sec2_reconstruction.tex}

\input{sec3_duality.tex}
\input{sec4_finite-dim-reduction.tex}

\input{secN_equivalence.tex}

\input{sec5_numerics.tex}

\input{sec6_comparison.tex}
\input{sec7_conclusion.tex}

\section*{Acknowledgements}

We gratefully acknowledge useful conversations with
Norman Christ and Sebastian Mizera.
RA is an Ernest Kempton Adams Postdoctoral Fellow supported in part by the Ernest Kempton Adams fund for Physical Research of Columbia University.
RA and SF are supported in part by U.S. DOE grant No. DE-SC0011941.
SF also is supported in part by
the National Science Foundation Graduate Research Fellowship Program under Grant No. DGE-2036197. Any opinions, findings, conclusions or recommendations expressed in this material are those of the author(s) and do not necessarily reflect the views of the National Science Foundation.
PO is supported in part by the U.S. Department of Energy, Office of Science, Office of Nuclear Physics under grant Contract Number DE-SC0012704 (BNL).

The numerical calculations in this work made use of \textsc{NumPy}~\cite{vanderWalt:2011bqk,Harris:2020xlr}, \textsc{FLINT}~\cite{flint}, SymPy~\cite{10.7717/peerj-cs.103}, \textsc{SciPy}~\cite{2020SciPy-NMeth}, and \textsc{mpmath}~\cite{mpmath}.
Figures were generated using \textsc{matplotlib}~\cite{Hunter:2007}. %

\appendix

\input{app_finite-dim.tex}

\input{app_primal-dual.tex}
\input{app_equivalence_proof.tex}

\input{app_np-moment.tex}
\input{app_hvp-kernel.tex}

\bibliography{main.bib}

\end{document}

%% file: abstract.tex
This work introduces the causal bootstrap, a framework for bounding smeared spectral observables from finite non-perturbative Euclidean data. 
The method optimizes over the convex set of positive spectral densities compatible with the data to compute rigorous upper and lower bounds on the target observable.
Statistical uncertainties, including correlations, are incorporated through compatibility regions using the covariance matrix.
The bounds are equivalent, via Lagrange duality, to certified bounds on
the target smearing kernel.
For polynomial, rational, and piecewise rational kernels, the resulting dual problems can be reduced to finite-dimensional semidefinite programs using techniques familiar, e.g., in the numerical conformal bootstrap.
The present formulation clarifies the relation to moment problems, Nevanlinna--Pick interpolation, and linear kernel-reconstruction methods.
Numerical examples demonstrate the method.

%% file: sec1_intro.tex
\section{Introduction}
\label{sec:introduction}

Extracting spectral information from Euclidean-time correlation functions is a central problem in lattice field theory.
In the simplest zero-temperature case, the full problem corresponds mathematically to inverting the Laplace transform
\begin{equation}
    \label{eq:laplace-transform}
    C(t) = \int_0^\infty dE \, \rho(E) e^{-E t},
\end{equation}
where $C(t)$ is a Euclidean-time connected correlation function and $\rho(E)$ is a spectral density.
Euclidean correlators can be computed non-perturbatively from first-principles lattice field theory using Monte Carlo methods.
These methods give $C(t)$ at a finite, discrete set of points with finite statistical uncertainty, rendering the reconstruction problem ill-posed.

In recent years significant progress has been made by avoiding direct reconstruction of the spectral function and instead targeting \emph{smeared} spectral quantities~\cite{Poggio:1975af,Barata:1990rn,Hansen:2017mnd,Bulava:2019kbi,Hansen:2019idp,Patella:2024cto,Juttner:2026fui}.
That is, rather than attempting to determine $\rho(E)$ pointwise, one chooses a kernel $K(E)$ and attempts to compute an integrated quantity of the form
\begin{equation}
    \label{eq:intro-smeared-observable}
    \mathcal{K}[\rho] = \int_0^\infty dE \, K(E) \rho(E).
\end{equation}
Many phenomenologically relevant hadronic observables can be recast as smeared spectral reconstruction problems for different choices of $K$, including contributions to the anomalous magnetic moment of the muon, inclusive cross sections, transport coefficients, and inclusive decay rates~\cite{Jeon:1995zm,Meyer:2007ic,Meyer:2011gj,Bailas:2020qmv,Gambino:2020crt,Evangelista:2023fmt,ExtendedTwistedMass:2024myu}.
Linear methods have been proposed to estimate such quantities from Euclidean data, including Backus--Gilbert~\cite{Backus:1968svk,Pijpers:1992} and the closely related approach of Hansen, Lupo and Tantalo (HLT)~\cite{Hansen:2019idp}, as well as methods based on Chebyshev polynomials~\cite{Bailas:2020qmv,Gambino:2020crt,Fukaya:2020wpp}, Gaussian processes~\cite{valentineI,valentineII,Horak:2021syv,Candido:2024hjt,DelDebbio:2024lwm} and the Mellin transform~\cite{Bruno:2024fqc}.
Nonlinear methods based on analyticity, moment problems, or convex optimization have also been proposed to obtain rigorous constraints on the allowed values of smeared spectral functions~\cite{Bergamaschi:2023xzx,PhysRevLett.126.056402,Abbott:2025snz,Fields:2025glg,Abbott:2026hfy,Lawrence:2024hjm}.

The central question addressed in this work is the following.
Given a finite set of Euclidean correlator data and imposing explicitly the positivity of the spectral density, ``What are the sharp upper and lower bounds on $\mathcal{K}[\rho]$?" 
This framing of the question turns spectral reconstruction into an optimization problem.
For noisy Monte Carlo data, compatibility means lying within an 
uncertainty ellipsoid (induced by the covariance matrix)
surrounding the measured correlator, leading to conservative confidence intervals for the target observable. 
In the limit of exact data, the ellipsoid shrinks and compatibility turns into an interpolation of the known correlator values.

The present work develops this point of view as a bootstrap method for spectral functions.
The term ``bootstrap'' is used here in the same broad sense as in the numerical conformal bootstrap~\cite{Simmons-Duffin:2016gjk}: a finite set of non-perturbative input data is combined with positivity constraints to produce rigorous bounds on quantities of interest.
This use of ``bootstrap'' is distinct from the statistical bootstrap, wherein independent data are resampled in order to assess statistical uncertainties.
The proposed method also differs from recent work on the quantum mechanical bootstrap (e.g., Refs.~\cite{Han:2020bkb,Lin:2024vvg,Cho:2025vws,Lawrence:2024mnj}) which incorporate positivity but do not leverage additional non-perturbative input data.

The first novel contribution of this work is the unification of several spectral reconstruction methods as a bootstrap procedure.
Bounds coming from methods such as Nevanlinna--Pick (NP) interpolation~\cite{Bergamaschi:2023xzx} and moment problems~\cite{Abbott:2025snz,Abbott:2026hfy} can be cast in the language of convex optimization problems discussed in Ref.~\cite{Lawrence:2024hjm}, and such reformulations yield mathematically equivalent tight bounds on reconstructed smeared spectral quantities when respectively applied to equivalent input data, in a sense to be explained later (cf. \cref{sec:equivalences}). 
These statements hold both for the scalar problems, as well as when generalized to Hermitian matrices of correlation functions. %
The matrix-valued formulation can improve bounds by leveraging additional operators, similar in spirit but wholly different in realization to the generalized eigenvalue problem (GEVP) approach to spectroscopy~\cite{Luscher:1990ck,Blossier:2009kd}.
Methods from convex optimization are used to constrain the spectral function
that produces the most extreme $\mathcal{K}[\rho]$ while remaining compatible with the observed data.
The linear methods discussed in Refs.~\cite{Gambino:2020crt,Bailas:2020qmv,Hansen:2019idp} are closely related to the dual formulation of the optimization problem pioneered in Ref.~\cite{Lawrence:2024hjm}.
The primary distinction is that the latter introduces upper and lower bounds on the kernel, rather directly targeting its controlled approximation.
For positive spectral functions, this immediately allows to infer rigorous bounds on the resulting smeared observable.

The second novel contribution of this work is to apply methods from the bootstrap literature in order to reduce the optimization problems discussed above to finite-dimensional semi-definite programs (SDPs)~\cite{Simmons-Duffin:2015qma,isii1964inequalities}.
In particular, the basis functions used to construct bounds based on
both Nevanlinna--Pick interpolation and moment problems admit a polynomial structure that allows any spectral reconstruction with a piecewise rational smearing kernel to be reduced to a mathematically equivalent finite-dimensional problem. 
This can then be solved efficiently using primal-dual interior point methods as in the numerical conformal bootstrap~\cite{Simmons-Duffin:2015qma}.
In contrast to the dual formulation discussed in Ref.~\cite{Lawrence:2024hjm}, such formulations are entirely finite-dimensional, meaning that there are finite number of variables and a finite number of constraints.

The remainder of this work is structured as follows.
\Cref{sec:spectral-reconstruction,sec:duality,sec:finite-diml-reduce} develop the primal formulation, dual certificates, and finite SDP reduction.
\Cref{sec:equivalences} establishes equivalences between certain families of bounds.
\Cref{sec:numerics} presents numerical demonstrations of the proposed methods.
\Cref{sec:comparison,sec:discussion} compare the method with related approaches and discuss practical assumptions for lattice applications.

Code for the numeric examples in \cref{sec:numerics} is available at \cite{code_github}.

%% file: sec2_reconstruction.tex
\section{Spectral Reconstruction Problems}\label{sec:spectral-reconstruction}

This section defines the positive spectral measures, correlator functionals, and target observables used throughout the paper.
These ingredients lead to primal optimization problems for exact and noisy data.
Scalar versions certain optimization problems have already appeared in the literature; Ref.~\cite{Lawrence:2024hjm} has discussed the scalar version of \cref{eq:opt-problem-def-exact-block,eq:opt-problem-def-inexact-block}.
To the best of the authors' knowledge, \Cref{eq:opt-problem-def-inexact-block-sdp} is novel to the present work.

\subsection{Spectral densities and smeared observables}

Consider a set of operators $\{\mathcal{O}_0,\dots,\mathcal{O}_{r-1}\}$ in a zero-temperature quantum field theory.
The associated matrix of two-point Euclidean correlation functions admits the spectral decomposition
\begin{align}
    C_{ab}(t)
    = \braket{\mathcal{O}_a(t) \mathcal{O}_b(0)^\dag}
    = \int_0^\infty dE \, \rho_{ab}(E) e^{-E t},
    \label{eq:spectral-rep-exp-Et}
\end{align}
with the spectral density
\begin{equation}
\label{eq:rho-initial-def}
    \rho_{ab}(E)
    = \sum_{n=0}^\infty
    \braket{0 | \mathcal{O}_a | n}
    \braket{n | \mathcal{O}_b^\dag | 0}
    \delta(E - E_n)
\end{equation}
where $\{\ket{n}\}$ denotes an energy eigenbasis of the physical Hilbert space.
From \cref{eq:rho-initial-def}, it is clear that $\rho$ is a positive semidefinite Hermitian matrix-valued spectral function;
This well-known positivity can be understood as a consequence of reflection positivity, and is the guiding physical constraint used throughout this work. 

For the present applications, the target is a smeared spectral quantity, obtained by integrating the spectral density against a chosen kernel.
For a Hermitian matrix-valued kernel $K(E)$, define
\begin{equation}
\label{eq:smeared-observable-energy}
    \mathcal{K}[\rho]
    = \int_0^\infty dE \, K(E) \matdot \rho(E),
\end{equation}
where $X \matdot Y = \Tr[X^\dag Y]$ is the Hermitian trace pairing; for Hermitian $X$ and $Y$, the trace pairing is $\Tr[XY]$ and real.
More generally, the aim of this work is to determine the minimum and maximum possible values of \cref{eq:smeared-observable-energy} among all positive spectral densities compatible with the available Euclidean data.

\subsection{General formulation}

It will be useful to package several superficially distinct spectral reconstruction problems in a common notation.
Throughout this formulation, matrix-valued quantities live in the space $\hermmat{r}$ of $r\times r$ Hermitian matrices.
Let $I \subset \R$ be an interval, let $r$ be a positive integer, and let $\mathcal{M}^r(I)$ denote the space of $r\times r$ Hermitian matrix-valued measures on $I$.
Let $\mathcal{M}^r_+(I)$ denote the corresponding space of positive-semidefinite spectral functions.
That is, $\rho\in\mathcal{M}^r_+(I)$ means that $\int_Idx\,  f(x)^\dag \rho(x) f(x)\geq0$ for every compactly supported vector-valued test function $f:I\to\C^r$.

Given basis functions $b_t:I\to\R$, define the correlator data generated by $\rho$ as
\begin{equation}
\label{eq:corr-def-general-bt}
    C_t[\rho] = \int_I dx \, b_t(x) \rho(x),
    \qquad t = 0,\dots,N_t-1.
\end{equation}
Throughout the present work, $C_t$ is assumed to be known on $t=0,1,\dots$, but the formulas straightforwardly generalize to any finite, discrete set.
Unless indicated otherwise, all quantities are given in units of the lattice spacing.
The left-hand side of \cref{eq:corr-def-general-bt} can be understood as a rank-3 tensor, $C_{tab}\equiv C_{ab}(t)$.
For the zero-temperature spectral representation in \cref{eq:spectral-rep-exp-Et}, the interval is $I=[0,\infty)$ and the basis functions are $b_t(E)=e^{-Et}$.
For a finite-temperature bosonic correlator with temporal extent $N_t$, the usual basis takes the form
\begin{equation}
    b_t(E)
    =
    \frac{\cosh[E(t-N_t/2)]}{\sinh(EN_t/2)}.
\end{equation}
The target observable is similarly written as
\begin{equation}
\label{eq:general-kernel-functional}
    \mathcal{K}[\rho]
    =
    \int_I dx\, K(x) \matdot \rho(x),
\end{equation}
where $K:I\to\hermmat{r}$ is the chosen smearing kernel.

The formulation in terms of basis functions also encompasses several analyticity-based approaches.
For instance, moment problems arise from the zero-temperature spectral representation after the change of variables $\lambda = e^{-E}$.
With $\lambda\in[0,1]$, \cref{eq:spectral-rep-exp-Et} becomes
\begin{equation}
\label{eq:corr-moment-rep-hausdorff}
    C_{ab}(t)
    =
    \int_0^1 d\lambda \,
    \tilde{\rho}_{ab}(\lambda) \lambda^t,
\end{equation}
where
$\tilde{\rho}(\lambda(E)) = \lambda(E)^{-1} \rho(E)$.
Thus the Euclidean correlator is the sequence of moments of a positive measure on $[0,1]$, with basis functions
\begin{equation}
    b_t^\text{Moment}(\lambda) = \lambda^t.
\end{equation}
In the same variables, a smeared observable can be written as
\begin{equation}
    \mathcal{K}[\rho]
    =
    \int_0^1 d\lambda \,
    K(-\log\lambda) \matdot \tilde{\rho}(\lambda),
\end{equation}
up to the same change-of-variables convention used in the definition of $\tilde{\rho}$.
This is the Hausdorff moment problem relevant for zero-temperature non-staggered correlators~\cite{Abbott:2025snz}.

Analyticity-based interpolation methods are naturally phrased in terms of the Stieltjes transform
\begin{equation}
\label{eq:stieltjes-trans-def}
    G(z)
    =
    \int_I dx\, \frac{\rho(x)}{x-z},
\end{equation}
where $z$ is a complex variable in the upper half plane and $I$ is an interval containing the support of $\rho$.
For a positive spectral density, $G$ is a Pick function\footnote{Pick functions are also equivalently called (Herglotz-)Nevanlinna functions in the literature.}: its imaginary part is positive semidefinite in the upper half plane.
Conversely, the positive density can be recovered from the boundary value of the imaginary part through the Stieltjes--Perron inversion formula,
\begin{equation}
    \label{eq:stieltjes-perron-inversion}
   \rho(x) = \frac{1}{\pi} \lim_{\epsilon_0 \to 0} G(x + i \epsilon_0)
\end{equation}
in the usual distributional sense. Together, \cref{eq:stieltjes-trans-def,eq:stieltjes-perron-inversion} provide a bijection between positive-semidefinite spectral functions and Pick functions; this is the key fact that allows analyticity-based approaches to be formulated in terms of a positive spectral function below.
The imaginary part of $G(z)$ at a point $z=x_0+i\epsilon_0$ is itself a smeared spectral function,
\begin{equation}
\label{eq:imag-G-eq-cauchy}
    \frac{1}{\pi}\imag G(x_0+i\epsilon_0)
    =
    \int_I dx \, \rho(x)
    K^\text{Cauchy}_{x_0,\epsilon_0}(x),
\end{equation}
where
\begin{equation}
\label{eq:cauchy-kernel-def}
   K^\text{Cauchy}_{x_0,\epsilon_0}(x)
   =
   \frac{1}{\pi}
   \frac{\epsilon_0}{(x-x_0)^2+\epsilon_0^2}.
\end{equation}

If the input data are values of $G(z_n)$ at fixed points $z_n$ in the upper half plane, their real and imaginary parts are again linear functionals of $\rho$.
Indeed, taking real and imaginary parts of \cref{eq:stieltjes-trans-def} gives the basis functions
\begin{align}
    b_n^{\text{NP},\real}(x)
    &=
    \frac{x-\real z_n}{(x-\real z_n)^2+(\imag z_n)^2},
    \label{eq:NP-basis-real} \\
    b_n^{\text{NP},\imag}(x)
    &=
    \frac{\imag z_n}{(x-\real z_n)^2+(\imag z_n)^2},
    \label{eq:NP-basis-imag}
\end{align}
where the basis functions are indexed by $n$ instead of $t$ to avoid confusion with Euclidean time.
Thus, analyticity-based interpolation can be incorporated into the same primal problems, with a different choice of basis functions.

Once the basis functions and target kernel have been fixed, the goal is to compute the allowed range of $\mathcal{K}[\rho]$, constrained by spectral positivity and compatibility with the available data.
The following two subsections make this statement precise, first for exact and then for noisy data.

\subsection{Exact data}

Suppose first that the correlator values $\hat{C}_t$ are known exactly for $t=0,\dots,N_t-1$.
The lower bound on $\mathcal{K}[\rho]$ is obtained by minimizing the target observable over all positive densities that reproduce these data:
\begin{optproblem}{Primal spectral reconstruction problem for exact data}
\label{eq:opt-problem-def-exact-block}
\minimize_{\rho \in \mathcal{M}^r_+(I)} \quad & \mathcal{K}[\rho]  \\
\st \quad & C_t[\rho] = \hat{C}_t \qquad \forall t \in \{0, \dots, N_t-1\}.
\end{optproblem}
The corresponding upper bound is obtained by maximizing $\mathcal{K}[\rho]$ subject to the same constraints, or equivalently by minimizing $-\mathcal{K}[\rho]$.
Thus, a complete spectral bound requires solving two optimization problems, one for the lower endpoint and one for the upper endpoint.

Following the nomenclature of convex programming, \cref{eq:opt-problem-def-exact-block} will be referred to as the primal problem.
A choice of $\rho$ is called \emph{primal feasible} if it satisfies the positivity constraint and reproduces the correlator data.
In the scalar case $r=1$, \cref{eq:opt-problem-def-exact-block} is an infinite-dimensional linear program.
For $r>1$, the  Hermitian positivity condition makes it an infinite-dimensional semidefinite program.
In the language of the moment-problem literature, the same problem is an instance of a generalized moment problem~\cite{bertsimas2005optimal,lasserre2008semidefinite}.

The infinite-dimensional nature of \cref{eq:opt-problem-def-exact-block} presents a practical difficulty.
An arbitrary positive spectral function cannot be represented directly on a computer, and a naive discretization of the spectral variable would introduce an additional approximation.
The purpose of \cref{sec:duality,sec:finite-diml-reduce} is to show that, for important classes of basis functions and kernels, this infinite-dimensional problem can instead be transformed into a finite-dimensional semidefinite program without such discretization.

\subsection{Inexact data}\label{sec:inexact-data}

For practical lattice field theory calculations, Euclidean correlation functions are computed with finite statistical uncertainty.
The interpolation constraint of \cref{eq:opt-problem-def-exact-block} must therefore be relaxed to one of statistical consistency with the measured correlator.

Let $\hat{C}$ denote the rank-3 tensor of measured correlator data, $\hat{C}_{tab}$.
The Hermitian, positive-definite
covariance matrix $\Sigma$ is a rank-6 tensor, $\Sigma_{tab,t'a'b'}$.
The covariance matrix induces a norm defined by
\begin{equation}
\label{eq:cov-metric-def}
    \norm{C}_\Sigma^2
    = C^\dag \Sigma^{-1} C
    = C^*_{tab} \Sigma^{-1}_{tab, t'a'b'} C_{t'a'b'}.
\end{equation}
For numerical applications, it is convenient to use suitably vectorized versions of these quantities, in which case $\hat{C}_{tab}$ is understood to contain $N_{\rm corr}$ independent components.
A correlator $C[\rho]$ is considered statistically consistent with $\hat{C}$ if
\begin{equation}
\label{eq:C-hatC-stat-conisistent}
    \norm{C[\rho] - \hat{C}}_\Sigma \leq \sigma_0,
\end{equation}
where $\sigma_0>0$ is a chosen cutoff. %
\Cref{eq:C-hatC-stat-conisistent} imposes a cutoff on the correlated $\chi^2$ deviation between the candidate correlator and the measured central value.
The cutoff may be chosen using an assumed Gaussian distribution, by a data-driven prescription, or by another procedure appropriate to the analysis~\cite{Christ:2024nxz,Bruno:2022mfy}.

The noisy-data version of the primal problem is thus:
\begin{optproblem}{Primal spectral reconstruction problem for noisy data, second-order cone form}
\label{eq:opt-problem-def-inexact-block}
\minimize_{\rho \in \mathcal{M}^r_+(I)} \quad & \mathcal{K}[\rho]  \\
\st \quad & \norm{C[\rho] - \hat{C}}_{\Sigma} \leq \sigma_0 .
\end{optproblem}
\noindent
As noted in Ref.~\cite{Lawrence:2024hjm}, if the consistency region in \cref{eq:C-hatC-stat-conisistent} contains the true correlator with a given probability, then
\cref{eq:opt-problem-def-inexact-block} gives bounds for $\mathcal{K}[\rho]$ at the same level of probability.
They are conservative in the sense that the optimization propagates the entire allowed region of correlators to the target observable.
The resulting interval may therefore be wider than the uncertainty obtained from any particular estimator of $\mathcal{K}[\rho]$.
Note that a solution to \cref{eq:opt-problem-def-inexact-block} need not always exist (e.g., for $\sigma_0=0$), since noisy data cannot be represented exactly by a positive spectral function.

\Cref{eq:opt-problem-def-inexact-block} imposes a second-order cone constraint.
For the finite-dimensional reduction below, it is useful to rewrite it as an equivalent semidefinite constraint.
Using the Schur complement, \cref{eq:C-hatC-stat-conisistent} can be rewritten as
\begin{equation}
\label{eq:Gamma-rho-def}
\Gamma[\rho]
\equiv
\begin{pmatrix}
    \Sigma & C[\rho] - \hat{C} \\
    (C[\rho] - \hat{C})^\dag & \sigma_0^2
\end{pmatrix}
\succeq 0,
\end{equation}
since $\Sigma$ is positive definite. 
The same noisy-data problem can therefore be written as
\begin{optproblem}{Primal spectral reconstruction problem for noisy data, SDP form}
  \label{eq:opt-problem-def-inexact-block-sdp}
\minimize_{\rho \in \mathcal{M}^r_+(I)} \quad & \mathcal{K}[\rho]  \\
\st \quad & \Gamma[\rho] \succeq 0 .
\end{optproblem}
The advantage of \cref{eq:opt-problem-def-inexact-block-sdp} is that the noisy-data constraint now has the same semidefinite character as the positivity condition on $\rho$.
This observation will allow parallel treatment of exact and noisy data in the dual and finite-dimensional formulations below.

%% file: sec3_duality.tex
\section{Duality}\label{sec:duality}

The primal problems in \cref{eq:opt-problem-def-exact-block,eq:opt-problem-def-inexact-block-sdp} are formulated directly in terms of the spectral density.
This section uses Lagrange duality to replace the optimization over positive spectral functions by a search for certified bounds on the target kernel.
The dual problems have finitely many optimization variables, and their pointwise positivity constraints will be reduced to finite SDPs in \cref{sec:finite-diml-reduce}.
Scalar versions of some optimization problems have already appeared in the literature; Ref.~\cite{Lawrence:2024hjm} has discussed the scalar version of \cref{eq:opt-problem-def-exact-block-dual}.
To the best of the authors' knowledge, \Cref{eq:opt-problem-def-inexact-block-dual,eq:opt-problem-inexact-cone-dual}, as well as the application of the methods in \cref{sec:recov-prim-solut} to problems in lattice field theory are novel to the present work.

\subsection{Dual certificates for exact data}\label{sec:lagrange-exact}

The basic idea is easy to motivate conceptually.
Suppose that matrices $g_t\in\hermmat{r}$ can be chosen such that
\begin{equation}
\label{eq:dual-certificate-condition}
    M_g(x)
    \equiv
    K(x) - \sum_t g_t b_t(x)
    \succeq 0
    \qquad \forall x\in I.
\end{equation}
For reasons that will become clear, $M_g$ will be referred to as the dual certificate.
For any positive spectral density $\rho\in\mathcal{M}^r_+(I)$ compatible with the exact data, positivity implies
\begin{align}
    0
    &\leq
    \int_I dx\, M_g(x)\matdot \rho(x) \nonumber \\
    &=
    \mathcal{K}[\rho]
    -
    \sum_t g_t \matdot C_t[\rho] \nonumber \\
    &=
    \mathcal{K}[\rho]
    -
    \sum_t g_t \matdot \hat{C}_t .
    \label{eq:dual-bound-direct}
\end{align}
Thus every choice of $g_t$ satisfying \cref{eq:dual-certificate-condition} provides a rigorous lower bound
\begin{equation}
    \mathcal{K}[\rho]
    \geq
    \sum_t g_t \matdot \hat{C}_t
\end{equation}
for every feasible spectral density.
In this sense, the dual variables define a one-sided reconstruction of the target kernel by the basis functions.
The inequality in \cref{eq:dual-certificate-condition}, rather than an approximate equality, is what turns this reconstruction into a certificate.

The best lower bound obtainable in this way is the solution of the dual optimization problem:
\begin{optproblem}{Dual spectral reconstruction, exact data}
\label{eq:opt-problem-def-exact-block-dual}
\maximize_{g_t \in \hermmat{r}} \quad & \sum_{t} g_{t} \matdot \hat{C}_{t} \\
\st \quad & K(x) - \sum_{t} g_{t} b_t(x) \succeq 0 ,
\quad \forall x \in I.
\end{optproblem}
A feasible point of \cref{eq:opt-problem-def-exact-block-dual} is called dual feasible.
The optimal value of the dual problem is always a lower bound on the optimal value of the primal problem, a statement known as weak duality.
Explicitly, if $g^*$ is dual optimal and $\rho_*$ is primal optimal, then
\begin{equation}
\label{eq:weak-duality}
    \mathcal{K}[\rho_*]
    \geq
    \sum_t g_t^* \matdot \hat{C}_t.
\end{equation}
Under mild regularity assumptions, this inequality is saturated, and the primal and dual optima agree.
This stronger statement, known as strong duality, is expected to hold in most physical applications by Slater's condition, which roughly requires the existence of a strictly feasible point~\cite{slater2013lagrange,Boyd:2004fnq}.
An instance where Slater's condition would not apply arises in analyticity-based interpolation methods when the input data specify the feasibly spectral density uniquely. 
However, such extremal problems occur with probability zero, and hence are not likely to cause issues in practice.

The same dual problem also follows from the standard Lagrange-duality construction.
Introduce a positive multiplier $\mu\in\mathcal{M}^r_+(I)$ for the positivity constraint and matrices $g_t\in\hermmat{r}$ for the interpolation constraints.
The Lagrange function is
\begin{align*}
    L&(\rho,\mu,g)\\
    &=
    \mathcal{K}[\rho]
    -
    \int_I dx\, \mu(x)\matdot \rho(x)
    -
    \sum_t g_t\matdot\left(C_t[\rho]-\hat{C}_t\right)
    \\
    \begin{split}
    &=
    \int_I dx\, \rho(x)\matdot
    \left[
        K(x)-\sum_t g_t b_t(x)-\mu(x)
    \right]\\
    & \quad +
    \sum_t g_t\matdot\hat{C}_t.
    \end{split}
\end{align*}
Minimization over the unconstrained variable $\rho$ is bounded from below only if
\begin{equation}
\label{eq:dualize-mu-reduce}
    \mu(x)
    =
    K(x)-\sum_t g_t b_t(x) \succeq 0, \quad \forall x\in I.
\end{equation}
Eliminating $\mu$ delivers precisely the dual problem, \cref{eq:opt-problem-def-exact-block-dual}.
This derivation is useful because it generalizes directly to the noisy-data problem below.

\subsection{Extremal feasible measures}\label{sec:recov-prim-solut}

The dual problem does not directly return the physical spectral density.
Instead, it identifies the most constraining one-sided approximation to the target kernel.
Nevertheless, when strong duality holds, the optimal dual solution contains useful information about the spectral density that saturates the bound.
This extremal density should not be interpreted as the true physical spectrum but rather as the feasible spectral density that produces the most extreme allowed value of the chosen observable.
Similar methods to those described in this section have also seen use in the numerical conformal bootstrap, referred to as the \emph{extremal functional method}~\cite{Simmons-Duffin:2016wlq,El-Showk:2012vjm}. 

Let $M_g(x)$ be the dual certificate defined in \cref{eq:dual-certificate-condition}.
At primal and dual optimum, saturation of \cref{eq:dual-bound-direct} requires
\begin{equation}
\label{eq:complementary-slackness}
    \rho_*(x) M_{g^*}(x)=0 ,\quad \forall x \in I
\end{equation}
in the sense of positive matrix measures.
In the convex-optimization literature, this condition is known as complementary slackness.
In the scalar case, it implies that $\rho_*$ can have support only where the dual certificate vanishes, i.e., $M_{g^*}(x)=0$.
If the zero set is a finite set $\{x_i\}_{i=1}^N$, the extremal feasible measure takes the form
\begin{equation}
\label{eq:rho-opt-compl-constrainted}
    \rho_*(x)
    =
    \sum_{i=1}^N A_i \delta(x-x_i),
    \qquad A_i\geq0.
\end{equation}
The support points $x_i$ are therefore determined by the zeros of the optimal dual certificate.
The coefficients $A_i$ are then constrained by the exact data,
\begin{equation}
\label{eq:rho-star-interpolate}
    C_t[\rho_*]
    =
    \sum_{i=1}^N A_i b_t(x_i)
    =
    \hat{C}_t,
\end{equation}
and by saturation of the bound,
\begin{equation}
\label{eq:rho-star-dual-match}
    \mathcal{K}[\rho_*]
    =
    \sum_{i=1}^N A_i K(x_i)
    =
    \sum_t g_t^* \matdot \hat{C}_t.
\end{equation}
Together these equations give finite-dimensional linear constraints on the nonnegative coefficients $A_i$.
These constraints can be solved using standard linear-algebra methods or non-negative least squares~\cite{lawson1995solving}.

For matrix correlators, complementary slackness is less restrictive.
A product of positive-semidefinite matrices can vanish even when neither factor vanishes.
Thus, in the matrix case, \cref{eq:complementary-slackness} constrains the null directions of $M_{g^*}(x)$ rather than simply localizing the measure to zeros of a scalar function.
Although the dual certificate is generically a full-rank object, $M_{g^*}(x)$ at the dual optimum is typically observed to be rank deficient {in the numeric examples discussed below.
The key point for the present work is that the extremal measure is a diagnostic of the bounding problem, not an independent extraction of the physical spectrum.

\subsection{Noisy data}\label{sec:lagrange-inexact}

A similar dual formulation extends to the noisy-data primal problem in \cref{eq:opt-problem-def-inexact-block-sdp}.
The only new feature is that compatibility with the data is encoded by the semidefinite constraint $\Gamma[\rho]\succeq0$ rather than by exact interpolation conditions.
The Lagrange multiplier for this constraint is a positive-semidefinite matrix $Z\in\hermmat{\Ncorr+1}$. 
Define, using \cref{eq:Gamma-rho-def} the quantity
\begin{equation}
\label{eq:Gamma0-def}
    \Gamma_0
    \equiv
    \Gamma[\rho=0]
    =
    \begin{pmatrix}
        \Sigma & -\hat{C} \\
        -\hat{C}^{\dag} & \sigma_0^2
    \end{pmatrix}.
\end{equation}
The entries of $Z$ conjugate to the off-diagonal block of $\Gamma[\rho]$ play the role of the matrices $g_t$ appearing in the exact-data dual.
More explicitly, this embedding can be written schematically as
\begin{equation}
\label{eq:Z-g-embed}
    Z
    =
    \frac{1}{2}
    \begin{pmatrix}
        \ast & g \\
        g^\dag & \ast
    \end{pmatrix},
\end{equation}
where the ordering of the components of $g$ matches the vectorization of the correlator data.

Dualizing \cref{eq:opt-problem-def-inexact-block-sdp} gives
\begin{optproblem}{Dual spectral reconstruction, noisy data}
\label{eq:opt-problem-def-inexact-block-dual}
\maximize_{Z \in \hermmat{\Ncorr+1}} \quad & - Z \matdot \Gamma_0\\
\st \quad
& K(x) - \sum_t g_t b_t(x) \succeq 0,
\quad \forall x \in I, \\
& Z \succeq 0,
\end{optproblem}
\noindent
with $g_t$ embedded in $Z$ as in \cref{eq:Z-g-embed}.
As in the exact-data case, the pointwise positivity condition is the kernel certificate.
The new matrix constraint $Z\succeq0$ encodes the covariance ellipsoid that defined statistical consistency in the primal problem.
Every feasible $Z$ provides a rigorous lower bound
\begin{equation}
    \mathcal{K}[\rho]
    \geq
    -Z\matdot\Gamma_0
\end{equation}
for every positive spectral density satisfying $\Gamma[\rho]\succeq0$.

The noisy-data dual has an equivalent form that makes contact with the exact-data certificate more directly.
Starting from \cref{eq:opt-problem-def-inexact-block}, one reformulates the constraint 
using a slack variable $y$:
\begin{align}
\norm{C[\rho]-\hat{C}}_\Sigma \leq \sigma_0    
\iff 
\begin{cases}
C[\rho]-\hat{C}-y = 0,\\
\norm{y}_\sigma \leq \sigma_0.
\end{cases}
\end{align}
The second line on the right-hand side constrains the slack variable $y$ to lie within an ellipsoid.
The associated Lagrange function is
\begin{multline}
    L(\rho,y;g,\mu) = \mathcal{K}[\rho]
    - \int \mu(x) \matdot \rho(x) \, dx \\
    - g \matdot \left( C[\rho] - \hat{C} - y\right).
\end{multline}
where $\mu \in \mathcal{M}^r_+(I)$ enforces the positivity of $\rho$.
Minimization over $\rho$ is bounded from below only when the dual certificate from \cref{eq:dual-certificate-condition} satisfies $M_g(x) = \mu(x) \succeq 0$ for all $x\in I$.
Minimization of $y$ over the ellipsoid gives
\begin{align}
    \minimize_{\norm{y}_\Sigma \leq \sigma_0} g \matdot y = - \sigma_0 \sqrt{g^\dagger \Sigma g}.
\end{align}
Combining these results gives the dual problem as a second-order cone problem~\cite{luo1996duality}:
\begin{optproblem}{Dual spectral reconstruction, noisy data, SOCP form}
\label{eq:opt-problem-inexact-cone-dual}
\maximize_{g_t \in \hermmat{r}} \quad
& \sum_t g_t \matdot \hat{C}_t
  - \sigma_0 \sqrt{g^\dag \Sigma g} \\
\st \quad
& K(x) - \sum_t g_t b_t(x) \succeq 0,
\quad \forall x \in I.
\end{optproblem}
The first term is the exact-data dual objective evaluated on the measured central values.
The second term penalizes directions in data space that are poorly determined by the covariance matrix.
In the limit $\sigma_0\to0$, this penalty vanishes and \cref{eq:opt-problem-inexact-cone-dual} reduces to the exact-data dual problem.
For the finite-dimensional SDP reduction in \cref{sec:finite-diml-reduce}, it will be more convenient to use the SDP form \cref{eq:opt-problem-def-inexact-block-dual}.
A similar dual formulation was reached in Ref.~\cite{Lawrence:2024hjm};
\cref{eq:opt-problem-inexact-cone-dual} is equivalent to Eq.~(28) therein after optimization over the variable Ref.~\cite{Lawrence:2024hjm} refers to as $\mu$.

The noisy-data problem also has complementary-slackness conditions, which are useful as diagnostics of extremal solutions.
For any primal feasible $\rho$ and dual feasible $Z$, the following equality holds:
\begin{equation}
\label{eq:inexact-manifest-slackness}
    \mathcal{K}[\rho]
    =
    -\Gamma_0\matdot Z
    +
    \Gamma[\rho]\matdot Z
    +
    \int_I dx\, \rho(x)\matdot M_g(x),
\end{equation}
where $M_g(x)=K(x)-\sum_t g_t b_t(x)$ is the dual certificate from \cref{eq:dual-certificate-condition}.
The last two summands on the right-hand side are nonnegative for feasible primal and dual points.
Thus equality at the optimum requires the following two conditions:
\begin{align}
    \rho_*(x)M_{g^*}(x) &= 0, \quad \forall x \in I\\
    \Gamma[\rho_*]Z_* &= 0.
    \label{eq:gamma_Z_0}
\end{align}
The former is the same complementary slackness condition that appeared for exact data.
The latter condition describes how the extremal correlator $C[\rho_*]$ sits on the boundary of the statistical consistency region.
To see this, write
\begin{equation}
    \Delta C[\rho]
    =
    C[\rho]-\hat{C}.
\end{equation}
Taking the Schur complement of $\Gamma[\rho]$ shows that its lower-right Schur complement is
\begin{equation}
\label{eq:inexact-schur-complement}
    \sigma_0^2
    -
    \Delta C[\rho]^\dag \Sigma^{-1}\Delta C[\rho] \geq 0.
\end{equation}
When the optimal dual matrix $Z_*$ is nonzero and $\Sigma$ is nonsingular, complementary slackness forces the inequality to be saturated in order to satisfy \cref{eq:gamma_Z_0}.
That is, the extremal correlator lies on the boundary of the allowed ellipsoid,
\begin{equation}
    \Delta C[\rho_*]^\dag \Sigma^{-1}\Delta C[\rho_*]
    =
    \sigma_0^2.
\end{equation}
This result makes manifest the fact that the primal optimizer will generically not reproduce the measured central values.
Indeed, it identifies the positive spectral density whose correlator is statistically consistent with the data while extremizing the target observable.

%% file: sec4_finite-dim-reduction.tex
\section{Finite-dimensional reduction}\label{sec:finite-diml-reduce}

The dual problems in \cref{eq:opt-problem-def-exact-block-dual,eq:opt-problem-def-inexact-block-dual} have finitely many optimization variables but still contain infinitely many constraints:
for every point $x$ in the spectral domain, the dual certificate in \cref{eq:dual-certificate-condition} must be positive semidefinite.
This section explains when those infinitely many positivity conditions can be replaced by a finite set of semidefinite constraints.
The key point is that polynomial matrix positivity conditions in a single variable can be represented as a finite sum of squares~\cite{hilbert1888darstellung}.
Whenever the dual constraints can be written as polynomial matrix inequalities, the spectral reconstruction problem can be reduced to a finite-dimensional semidefinite program.

The formulation of spectral reconstruction problems \cref{eq:pmp-reduce-dual-moment,eq:pmp-reduce-dual-moment-rational,eq:pmp-reduce-NP-problem,eq:inexact-pmp-sdp-direct} as polynomial matrix problems that can be reduced to finite semidefinite programs is a central result of the present work.

\subsection{Polynomial matrix programs}

A polynomial matrix program (PMP) is a semidefinite program of the form~\cite{Simmons-Duffin:2015qma}
\begin{optproblem}{General polynomial matrix program}
\label{eq:PMP-def}
\maximize_{y \in \R^n} \quad & b \cdot y\\
\st \quad &
    M^0_j(x) + \sum_{k=1}^n y_k M^k_j(x) \succeq 0, \\
    & \forall x \in I, \quad j=1,\dots,J.
\end{optproblem}
Here $I\subseteq\R$ is an interval, the variables are $y_k\in\R$, and the matrices $M^k_j(x)$ are matrix-valued polynomials in $x$.
\Cref{eq:PMP-def} contains an affine matrix-valued constraint on the $y_k$.
The same structure appears in spectral reconstruction whenever the target kernel and basis functions are polynomial, or can be made polynomial by multiplying by a positive denominator.

A simple example is the exact-data moment problem.
On an interval $I$ (e.g. $I = [0, 1])$ and with basis functions $b_t(\lambda)=\lambda^t$, the exact dual problem is
\begin{optproblem}{Dual moment problem}
\label{eq:pmp-reduce-dual-moment}
\maximize_{g_t \in \hermmat{r}} \quad & \sum_t g_t \matdot \hat{C}_t \\
\st \quad &
    K(\lambda) - \sum_t g_t \lambda^t
        \succeq 0
        \quad \forall \lambda \in I.
\end{optproblem}
\noindent
If $K(\lambda)$ is a polynomial matrix, then \cref{eq:pmp-reduce-dual-moment} is already a PMP after expanding the Hermitian matrices $g_t$ in a basis.
Let $\{T_\alpha\}$ be a basis for $\hermmat{r}$, with $\alpha=1,\dots,r^2$, and write
\begin{equation}
\label{eq:g-t-basis-decompose}
    g_t = \sum_\alpha y_{t\alpha} T_\alpha.
\end{equation}
Substituting \cref{eq:g-t-basis-decompose} into \cref{eq:pmp-reduce-dual-moment} gives a polynomial matrix inequality that is affine in the variables $y_{t\alpha}$, and hence has the form of \cref{eq:PMP-def}.

Many kernels of physical interest are more naturally represented as rational functions than as polynomials.
This structure is still compatible with the PMP reduction, provided the denominator is positive on the interval where the constraint is imposed.
For example, suppose that in \cref{eq:pmp-reduce-dual-moment} the kernel takes the form
\begin{equation}\label{eq:Kernel_lambda_PMP}
    K(\lambda)
    =
    \frac{n(\lambda)}{d(\lambda)} e_0 e_0^\top,
\end{equation}
where $n,d\in\R[\lambda]$, $d(\lambda)>0$ for all $\lambda\in I$, and $e_0=(1,0,\dots,0)$.
Such kernels arise naturally when reconstructing a smeared spectral function for the first interpolating operator $\mathcal{O}_0$, with the remaining operators serving to capture additional spectral information.
The contributions of other operators can be isolated similarly.
Multiplying the dual constraint by $d(\lambda)$ preserves the order $\succeq$ and gives the equivalent PMP
\begin{optproblem}{Dual moment problem, rational kernel}
\label{eq:pmp-reduce-dual-moment-rational}
\maximize_{g_t \in \hermmat{r}} \quad & \sum_t g_t \matdot \hat{C}_t \\
\st \quad &
    n(\lambda)e_0e_0^\top
    -
    \sum_t g_t \lambda^t d(\lambda)
        \succeq 0
        \quad \forall \lambda \in I.
\end{optproblem}
\noindent

Exact knowledge of $G(z)$ as in \cref{eq:imag-G-eq-cauchy} instead allows a similar treatment for a different class of kernels, namely those which are rational in $E$, rather than in $\lambda=e^{-E}$.
Using the basis functions in \cref{eq:NP-basis-real,eq:NP-basis-imag}, the exact-data dual has the form
\begin{optproblem}{Nevanlinna--Pick-based interpolation, rational form}
\label{eq:pmp-reduce-NP-problem}
\maximize_{g_n^{\real},g_n^{\imag}} ~
& \sum_n
\left(
g_n^{\real} \matdot \real \hat{G}_n
+
g_n^{\imag} \matdot \imag \hat{G}_n
\right) \\
\st ~
& K(E)
-
\sum_n
\tfrac{(E-\real z_n)g_n^{\real} + (\imag z_n) g_n^{\imag}}
{(E-\real z_n)^2+(\imag z_n)^2}
\succeq 0,\\
~ & \forall E\in I.
\end{optproblem}
When $\imag z_n>0$, each denominator is positive on the real axis.
Multiplying by the product of the denominators therefore converts the constraint into a polynomial matrix inequality, provided the target kernel is polynomial or rational with a compatible positive denominator.
The price of this exact conversion is that the polynomial degree can become large as the number of interpolation points grows.

For kernels that are not polynomial or rational in the variable used by the basis functions, one may approximate the target kernel by a polynomial or rational function.
In fact the Stone--Weierstrass theorem, which says that a continuous function can be approximated uniformly by a polynomial, suggests that such approximations will often converge rapidly.
Regardless, such an approximation changes the target observable and should be treated as a separate source of systematic uncertainty.
The finite-dimensional reduction is exact for the approximating kernel, but not automatically for the original kernel.
In applications, the resulting error must therefore be separately bounded or otherwise included in the final uncertainty budget.

In contradistinction to existing linear methods like Refs.~\cite{Hansen:2019idp,Gambino:2020crt}, the degree of approximation in the current proposal is \emph{not} directly tied to the number of input data used.
However, for a fixed-order polynomial approximation to a given kernel, the \emph{strength} of the bounds coming from the dual certificate (\cref{eq:dual-certificate-condition}) is generically expected to vary strongly depending on the number and quality of the input data.

\subsection{Finite SDP reduction}

Once the dual problem has been written as a PMP, the remaining task is to replace the constraint in \cref{eq:PMP-def}, $M(x)\succeq0$ for all $x\in I$, by a finite number of semidefinite constraints.
For polynomial matrices of a single variable, this reduction is possible using sum-of-squares representations.
The precise form depends on whether $I$ is the full real line, a half-line, or a bounded interval.
The details are collected in \cref{app:finite-dim}.

Conceptually, the reduction works as follows.
Choose a polynomial basis $\{q_i(x)\}$ and form the matrix of products $Q_{ij}(x)=q_i(x)q_j(x)$.
Positivity of a polynomial matrix $M(x)$ on an interval can be represented by positive-semidefinite coefficient matrices $Y_a$ such that
\begin{equation}
\label{eq:sdp_reduction_of_positivity}
    M(x)
    =
    \text{a finite linear expression in } Y_a \text{ and } Q(x).
\end{equation}
Equating this expression to the polynomial matrix appearing in the PMP gives finitely many linear equations in the variables $y_k$ and the positive matrices $Y_a$.
The result is an ordinary finite-dimensional SDP.

For the exact-data problem, the final SDP has the following structure:
\begin{equation}
\label{eq:dual-pmp-sdp-exact}
\begin{aligned}
    \maximize_y & \quad b\cdot y \\
    \st \quad &
    \Tr[A_pY] + (By)_p = c_p,
    \quad p=1,\dots,P, \\
    & Y\succeq0.
\end{aligned}
\end{equation}
The matrices $A_p$, the linear map $B$, and the constants $c_p$ are determined directly by the polynomial identities used to enforce positivity.
Values for $A_p$, $B$, and $c_p$ can be found either by matching polynomial coefficients between \cref{eq:sdp_reduction_of_positivity} and \cref{eq:dual-pmp-sdp-exact} or by evaluating the polynomial identity at sufficiently many sample points.
Additional information on obtaining $A_p$, $B$, and $c_p$ is given in Sections~2.2 and 2.5.2 of Ref.~\cite{Simmons-Duffin:2015qma}.
Different interval choices and block decompositions change these matrices, but not the overall structure of the problem.

\subsection{Extensions: noisy data and piecewise kernels}\label{sec:thresholds}

The noisy-data dual problem in \cref{eq:opt-problem-def-inexact-block-dual} fits into the same framework.
The pointwise positivity constraint on the dual residual $K(x)-\sum_t g_t b_t(x)$ is reduced exactly as in the exact-data problem.
The only additional ingredient is the semidefinite matrix variable $Z$ that encodes the covariance ellipsoid.
After the PMP constraints are reduced to finite form, the resulting SDP contains both the sum-of-squares matrix variable $Y$ and the noisy-data matrix variable $Z$.
The result is:
\begin{optproblem}{Dual spectral reconstruction, finite SDP}
\label{eq:inexact-pmp-sdp-direct}
\maximize_{Z} \quad & -Z\matdot\Gamma_0 \\
\st \quad
& \Tr[A_pY] + (By)_p = c_p,
\quad p=1,\dots,P, \\
&
    \Tr[\hat{A}_{t\alpha}Z]-y_{t\alpha}=0,
\\
& 
    t=0,\dots,N_t-1,\quad
    \alpha=1,\dots,r^2,
\\
& Y\succeq0,\quad Z\succeq0.
\end{optproblem}
\noindent
The $y_{t\alpha}$ are the coefficients of $g_t$ in a real basis of $\hermmat{r}$, while the $\hat{A}_{t\alpha}$ isolate the entries of $Z$ as in \cref{eq:Z-g-embed}.
More explicitly, the matrices $\hat{A}_{t\alpha}$ can be defined by the relation
\begin{align}
    \sum_\alpha \Tr \left[ \hat{A}_{t\alpha} \begin{pmatrix}
        * & g \\ g^\dagger & *
    \end{pmatrix}\right] T_{ab}^\alpha = 2 g_{t,ab}.
\end{align}

The matrices $Y$ and $Z$ may also be combined into a larger block-diagonal semidefinite variable $Y' = \diag(Y, Z)$.
With a corresponding block-diagonal objective matrix $D = \diag(0, -\Gamma_0)$, the noisy-data problem fits the more general SDP form
\begin{equation}
\label{eq:dual-pmp-sdp-full}
\begin{aligned}
    \maximize_{y, Y'} \quad & D\matdot Y' \\
    \st \quad &\Tr[A'_pY']+(B'y)_p=c'_p,
    \quad p=1,\dots,P',\\
    &Y'\succeq0,
\end{aligned}
\end{equation}
where $A'_p, B', c_p'$ are appropriate combinations of the un-primed variables above.
In implementation, the block structure of $Y'$ should be preserved rather than treated as a dense matrix.
This is the same kind of block structure exploited in conformal-bootstrap SDP solvers.
The block structure may be important computationally but does not change the bound.

The same logic also handles kernels that are polynomial or rational only piecewise, which is particularly useful for observables with kinematic thresholds.
For example, inclusive decays often involve kernels of the form
\begin{equation}
\label{eq:threshold-kernel-example}
    K_{\text{thresh}}(E)
    =
    \theta(M-E)\tilde{K}(E),
\end{equation}
where $\theta(\cdot)$ is a Heaviside step function, $M$ is a threshold energy, and $\tilde{K}(E)$ is polynomial or rational on the relevant interval.
Instead of approximating the step function, the dual positivity condition can be split across the intervals on which the kernel has different analytic forms.
For $I=[0,\infty)$, the dual constraint
\begin{equation}
\label{eq:pmp-dual-threshold}
    K_{\text{thresh}}(E)-\sum_t g_t b_t(E)\succeq0
    \qquad \forall E\in[0,\infty)
\end{equation}
is equivalent to the pair of constraints
\begin{subequations}
\label{eq:pmp-dual-threshold-expanded}
\begin{align}
    \tilde{K}(E)-\sum_t g_t b_t(E)
    &\succeq0,
    \qquad E\in[0,M],
    \label{eq:pmp-dual-threshold-expanded-a}
    \\
    -\sum_t g_t b_t(E)
    &\succeq0,
    \qquad E\in[M,\infty).
    \label{eq:pmp-dual-threshold-expanded-b}
\end{align}
\end{subequations}
Each interval produces its own sum-of-squares block in the finite SDP.
Thus piecewise polynomial or rational kernels can be handled exactly at the level of the dual constraints, provided each piece is polynomial or rational in the variable used for the PMP reduction.

After the PMP reduction, the remaining numerical task is to solve a finite-dimensional SDP.
This work follows the primal-dual interior point strategy used in conformal-bootstrap calculations and in SDPB~\cite{Simmons-Duffin:2015qma,Landry:2019qug}.
Such methods solve primal and dual finite SDPs simultaneously.
The corresponding duality gap is computable and provides a certificate that the finite SDP has been solved to the desired precision.

%% file: secN_equivalence.tex
\section{Equivalence of bounds between different formulations \label{sec:equivalences}}

Once seemingly different problems have been expressed in the common language of convex optimization, it becomes clear their formulations are in fact equivalent.
This equivalence is a central result of the present work. 

For instance, let $\{\hat{C}_t\}$ be a set of moments specifying a nondegenerate truncated Hamburger moment problem with $I = \R$ and $b_t(\lambda) = \lambda^t$, and consider the problem of determining the \emph{Wertevorrat} defined by
\begin{equation}
\label{eq:wertevorrat-defn}
    \Delta_n(z) = \{ G(z) \mid C_t[\rho] = \hat{C}_t \},
\end{equation}
$G(z)$ is
the Stieltjes transform of spectral function $\rho$, cf. \cref{eq:stieltjes-trans-def}
Using the methods derived from the moment problem and Nevanlinna--Pick literature discussed in Refs.~\cite{Bergamaschi:2023xzx,Abbott:2025snz,Abbott:2026hfy}, it is possible to determine $\Delta_n(z)$ to be a circle with analytically calculable center and radius as functions of the known moments $\hat C_t$ only.
The derivation of the Wertevorrat is deeply steeped in the analytic properties of $G(z)$, and
it is therefore remarkable that the same information can be accessed via a SDP formulation that references only the positivity of the spectral function with a given support.

To see this fact concretely, consider for simplicity a scalar $G(z)$ and define the following family of linear functionals on $\rho$ parameterized by $\theta \in [0, 2\pi)$
\begin{gather}
    \mathcal{K}_\theta[\rho]
    = \real \left[ e^{i\theta} G(z) \right]
    = \int_I d\lambda \,
    \rho(\lambda) K_\theta(\lambda) , \\
    \label{eq:K_theta}
    K_\theta(\lambda) = \real \left(
        \frac{e^{i\theta}}{z-\lambda}
    \right).
\end{gather}
The parameter $\theta$ allows one to sweep over the real and imaginary parts of $G(z)$ in order to map out the complete allowed region.
Examples of equivalent bounds on $G(z)$ can now be stated precisely.

\begin{lemma}[Equivalence of moment bounds]
\label{example:moment_bounds}
Consider a non-extremal truncated scalar Hamburger moment problem ($\lambda \in \R$) with moments $\hat{C}_t$ where $t=0, \dots, N=2n$.
For a given fixed $z$ in the upper half-plane, $G(z)$ is constrained to lie within a convex set
\begin{align}
    \Delta_{N}(z) = \{ G(z) | \rho\geq 0~\text{has moments}~\hat{C}_0, \dots, \hat{C}_N \}
\end{align}
The following two descriptions for $\Delta_{N}(z)$ are equivalent:
\begin{enumerate}
    \item $\Delta_N(z)$ is a disk with known analytic functions for the center $c_N(z)$ and radius $r_N(z)$
    in terms of the input moments $\hat{C}_t$, given explicitly in Refs.~\cite{Kovalishina1984,Abbott:2025snz}.
    \item $\Delta_{N}(z)$ is the intersection of the bounds obtained from solutions of the semidefinite programs for $\theta \in [0,2\pi)$ 
    \begin{align}
    \label{eq:sdp_wertevorrat_ct}
    \begin{split}
        \minimize_{\rho \in \mathcal{M}_+(\R)}
        \quad & \mathcal{K_\theta}[\rho]\\
        \st \quad & C_t[\rho] = \hat{C}_t, \quad t \in 0,\dots N.
    \end{split}
    \end{align}
\end{enumerate}
\end{lemma}
\noindent
A proof of \cref{example:moment_bounds} is sketched in \cref{app:equivalence_proof}.
For matrix-valued problems, the set $\Delta_N(z)$ is generically a Weyl matrix ball.
Similar results hold for moment problems over restricted intervals provided the interval is imposed in both setups.

The equivalences also directly extend to the dual formulation. For instance, Nevanlinna--Pick interpolation can be reframed as a dual optimization problem as follows.

\begin{lemma}[Equivalence of interpolation bounds]
\label{example:interpolation_bounds}
Consider a non-extremal complex interpolation problem in which $G(z)$ is specified at a finite set of interpolation nodes $\{z_n\}$ in the upper half-plane, $G(z_n)=\hat{G}_n$ for $n=0, \dots, N$.
For a given $z \notin \{z_n\}$ in the upper half-plane, $G(z)$ is constrained to lie within a convex set \begin{align}
    \Delta_{N}(z) = \{ G(z) | \rho \in \mathcal{M}^1_+([0, \infty)),~G(z_n)=\hat{G}_n \}.
\end{align}
The following descriptions for $\Delta_N(z)$ are equivalent:
\begin{enumerate}
    \item $\Delta_N(z)$ is the intersection of the two disks obtained by solving the Nevanlinna--Pick interpolation problem on the nodes $\{z_n\}$ for the functions $G(z)$ and $z G(z)$~\cite{krein1998interpolation}.
    Explicit analytic formulas for $\Delta_N(z)$ are given in Ref.~\cite{Bergamaschi:2023xzx} and the references therein.
    \item $\Delta_N(z)$ is the interior of the intersection of bounds obtained from solutions of the semidefinite programs for $\theta\in [0, 2\pi)$
    \begin{align}
    \begin{split}
    \maximize_{g_n^{\real},g_n^{\imag}} ~
    & \sum_n
    \left(
        g_n^{\real} \real \hat{G}_n
        + g_n^{\imag} \imag \hat{G}_n
    \right) \\
    \st ~ & K_\theta(E)
    - \sum_n
        \tfrac
        {(E-\real z_n)g_n^{\real} + (\imag z_n) g_n^{\imag}}
        {(E-\real z_n)^2+(\imag z_n)^2}
        \succeq 0,\\
        &\quad \forall E\in [0,\infty).
    \end{split}
    \end{align}
\end{enumerate}
\end{lemma}
\noindent

The proof for \cref{example:interpolation_bounds} is
very similar to that of \cref{example:moment_bounds} and thus omitted.
Similar equivalences are expected to hold for arbitrary intervals $I\subseteq \R$ provided the restriction to $I$ is included consistently on both sides.
For the bounds coming moment problems, the restriction is generically expected to be phrased in terms of intersections of Weyl matrix balls as in Ref.~\cite{heide2024weyl}.

Note that the SDP and analyticity-based formulations here provide complementary approaches to the same problem:
the analyticity-based approaches do not require an optimization step, and directly give the full structure of the Wertevorrat, whereas the SDP formulation readily generalizes to any polynomial or rational kernel and to noisy data.
Which approach is more useful for a given calculation should be determined on a case-by-case basis.

%% file: sec5_numerics.tex
\section{Numerical results}\label{sec:numerics}

This section presents numerical demonstrations of the bounds developed above.
The examples below use a custom Python implementation with arbitrary-precision arithmetic from \texttt{mpmath}~\cite{mpmath}.
Implementation details, including the SDP forms and interval positivity certificates, are given in \cref{app:finite-dim,sec:primal-dual}.

\subsection{Toy spectral density and input data}\label{sec:toy-model}

The toy spectral density is chosen to have a nontrivial matrix structure while remaining simple enough to evaluate exactly.
For $a,b\in\{0,1\}$, define
\begin{equation}
\label{eq:rho-toy-defn}
    \rho^{\text{toy}}_{ab}(E)
    =
    \sum_{k=0}^{N_\text{state}-1}
    Z_{ka}^* Z_{kb}
    \delta(E-E_k),
\end{equation}
with evenly spaced energies
\begin{equation}
    E_k = \frac{1}{10}(k+1).
    \label{eq:toy_spectrum}
\end{equation}
and $N_\text{state}=96$.
The overlap factors are
\begin{align}
Z_{n0}
&=
\begin{cases}
+\tfrac{1}{5} & n=0,\\
+1 & n \text{ odd},\\
-\tfrac{1}{5} & n>0 \text{ even},
\end{cases}
&
Z_{n1}
&=
\begin{cases}
1 & n \text{ even},\\
\tfrac{1}{5} & n \text{ odd}.
\end{cases}
\end{align}
These overlap factors are chosen to mirror a situation where the operator of interest has a relatively low overlap onto the ground state.
This is a common source of difficulty in spectral applications, since the target observable may be sensitive to low-energy states that are not cleanly isolated by the target operator.

The Euclidean correlator generated by \cref{eq:rho-toy-defn} is
\begin{equation}
\label{eq:toy-correlator}
    C_{t,ab}
    =
    \sum_{k=0}^{N_\text{state}-1}
    Z_{ka}^*Z_{kb} e^{-E_k t}.
\end{equation}
In all examples below, the correlator is sampled at $N_t=20$ time slices.
The scalar-correlator bounds use only the $C_{00}(t)$ component.
The matrix-correlator bounds use the full $2\times2$ correlator matrix, while the target observable is still chosen to involve the same component of the spectral density as in the scalar case.
This makes the scalar and matrix calculations directly comparable: the target is fixed, while the amount of positive-semidefinite input data is varied.

For moment-problem inputs, the data are the Euclidean correlator values in \cref{eq:toy-correlator}.
For inputs to Nevalinna--Pick-based interpolation, we instead use the Stieltjes transform of the same toy density,
\begin{equation}
\label{eq:toy-stieltjes}
    G_{ab}(z)
    =
    \sum_{k=0}^{N_\text{state}-1}
    \frac{Z_{ka}^*Z_{kb}}{E_k-z}.
\end{equation}
The transform is evaluated on the imaginary axis at $z=i\omega_n=2\pi i n/N_t$ for $n=0,\dots,N_t/2-1$. 
For numerical convenience, the point at zero frequency is shifted slightly to $\omega_0=10^{-3}\omega_1$, avoiding a point on the boundary of the upper half plane.
The moment-problem and Nevalinna--Pick-interpolation calculations therefore use different finite sets of input functionals of the same underlying toy density. 

To model statistical uncertainty, the covariance matrix is taken to be
\begin{equation}
\label{eq:toy-covariance}
    \Sigma_{tab,t'a'b'} = \covscale^2 
    \left[
        s+(1-s)\delta_{tt'}\delta_{aa'}\delta_{bb'}
    \right]
    \Sigma_{tab,t'a'b'}^{(0)}
\end{equation}
where $\covscale$ is an overall scale parameter, $s=1/2$ is a shrinkage parameter, and $\Sigma^{(0)}$ is defined by
\begin{equation}
    \Sigma_{tab,t'a'b'}^{(0)} =
    \hat{C}_{tab}\hat{C}_{t'a'b'} e^{-\gamma|t-t'|},
\end{equation}
with $\gamma^{-1}=1.3$.
For each numeric example, $\hat{C}_t$ is 
is drawn from a multivariate normal distribution
$\hat{C} \sim \mathcal{N}(C, \Sigma)$.

For simplicity, the covariance matrix is taken to have the same structure in the time domain and in the upper half-plane.
The cutoff $\sigma_0$ (cf.\cref{eq:C-hatC-stat-conisistent}) is taken to be
\begin{equation}
\sigma_0 = \sqrt{2 N_\text{corr}} \covscale,
\label{eq:cutoff_scaling}
\end{equation}
which roughly corresponds to a cutoff of 2 on the $\chi^2$ per ``degree of freedom."

\subsection{Spectral Reconstruction with Cauchy kernel}

This section presents results for a Cauchy-smeared spectral reconstruction,
\begin{align}
    \tilde{\rho}_\epsilon(E) 
    &= \int d\lambda\, \tilde{K}^{\rm Cauchy}_{E,\epsilon}(\lambda) \tilde{\rho}(\lambda) \\
    \tilde{K}^\text{Cauchy}_{E, \epsilon}(\lambda) &= \frac{\epsilon}{(\lambda - e^{-E})^2 + \epsilon^2}
    \label{eq:cauchy-defn-moment-numerics}
\end{align}
where $\tilde{\rho}(\lambda) = \lambda \rho(-\log \lambda)$ is the $\lambda$-space spectral function as in \cref{eq:corr-moment-rep-hausdorff}.
Since $\tilde{K}^\text{Cauchy}_{E, \epsilon}$ is a rational function of $\lambda$, the problem of bounding $\rho_\epsilon(E)$ can reformulated exactly as a finite-dimensional SDP via \cref{eq:pmp-reduce-dual-moment-rational}.
By \cref{eq:imag-G-eq-cauchy}, the parameter $\epsilon$ corresponds geometrically to a fixed distance above the real line in the complex-$\lambda$ plane.

\cref{fig:cauchy-sweep} shows the resulting bounds from solving \cref{eq:pmp-reduce-dual-moment-rational}.
The reconstruction used the toy-model inputs for $\hat{C}_t$ 
(cf. \cref{sec:toy-model} above)
on $N_t=20$ total timeslices.
A fixed value of $\epsilon=0.1$ was used for the choice of the smearing kernel.
As a stand-in for precise but realistic Monte Carlo data, the noise level of the input data is taken to be $\covscale = 10^{-4}$.

\cref{fig:cauchy-sweep} shows that the smeared spectral function $\tilde{\rho}_\epsilon(E)$ is bounded with good precision.
As expected, the bounds from the matrix-valued input data are generally stronger than from scalar-valued inputs.\footnote{The fact that the matrix bounds are not \emph{always} better is a consequence of choosing to scale the cutoff $\sigma_0$ with $N_{\rm corr}$, which is different in the scalar and matrix cases.
For fixed cutoff $\sigma_0$, the matrix bound will always be at least as tight as the scalar bound.
}
More generally, \cref{fig:cauchy-sweep} demonstrates how these bounds can be used with noisy input data to constrain energy-dependent smeared spectral quantities. 
As the preceding sections show, the convex-optimization formulation offers the freedom to choose kernels besides \cref{eq:cauchy-defn-moment-numerics}.

\begin{figure}
    \centering
    \includegraphics[width=\linewidth]{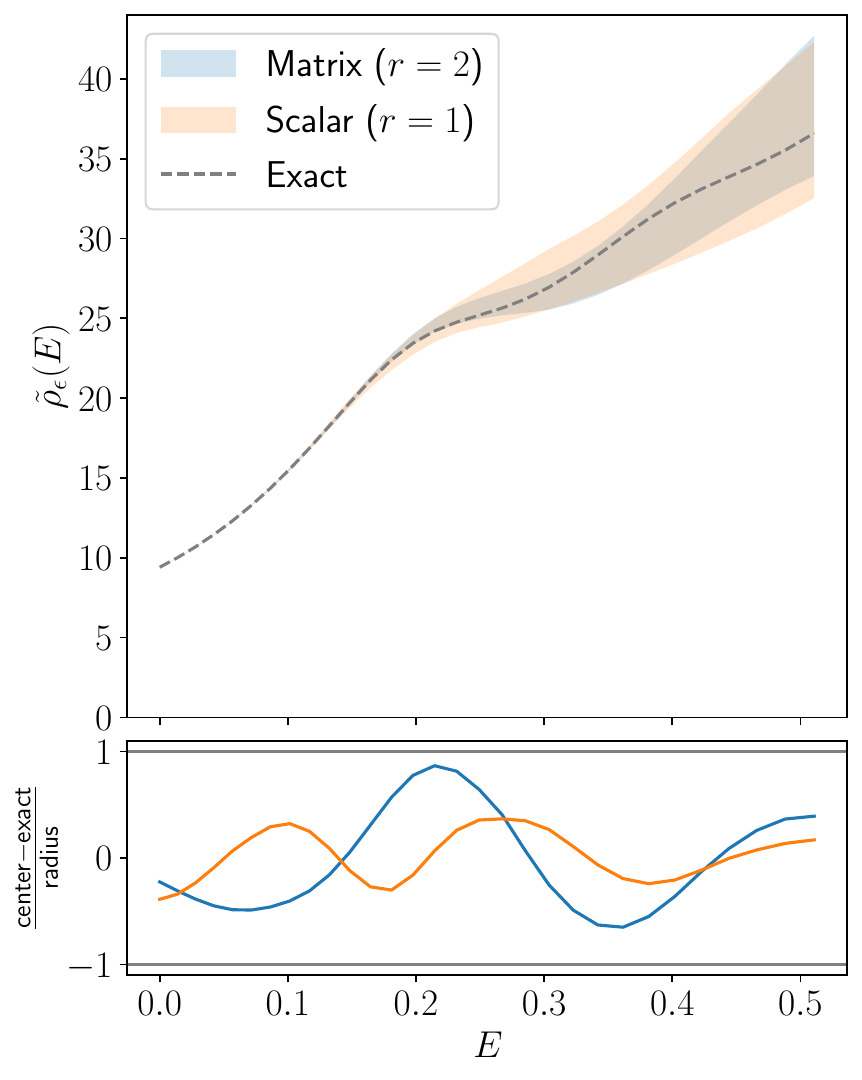}
    \caption{Bounds on the smeared spectral function with kernel $\tilde{K}^\text{Cauchy}$ as defined in \cref{eq:cauchy-defn-moment-numerics}.
    }
    \label{fig:cauchy-sweep}
\end{figure}

\subsection{Toy model for inclusive $\tau$ decay}\label{sec:tau-decay}

Inclusive hadronic decays of the $\tau$ lepton provide a useful example of a threshold spectral observable.
They are phenomenologically interesting because they can constrain CKM matrix elements and because persistent tensions with CKM unitarity have motivated renewed lattice interest; see Refs.~\cite{HeavyFlavorAveragingGroupHFLAV:2024ctg,FlavourLatticeAveragingGroupFLAG:2024oxs} and references therein.
The inclusive hadronic decay rate can be written in terms of longitudinal and transverse spectral densities as~\cite{Evangelista:2023fmt,ExtendedTwistedMass:2024myu}
\begin{equation}
\label{eq:R_tau}
    \frac{R^{(\tau)}_{uq}}{|V_{uq}|^2}
    \propto
    \sum_{X=T,L}
    \int_0^\infty dE\,
    K_X\!\left(\frac{E}{m_\tau}\right)
    E^2 \rho_{X,uq}(E),
\end{equation}
where $q=d,s$ and the omitted prefactors are known.
The kernels are
\begin{align}
    K_L(x)
    &=
    \frac{1}{x}(1-x^2)^2\theta(1-x),
    \\
    K_T(x)
    &=
    (1+2x^2)K_L(x).
\end{align}
The step function encodes the kinematic threshold at $E=m_\tau$.

For the numerical demonstration, we isolate a single channel and define
\begin{equation}
\label{eq:RX-tau-defn}
    R_X
    =
    \int_0^\infty dE\,
    K_X\!\left(\frac{E}{m_\tau}\right)
    E^2 \rho_X(E),
    \qquad X\in\{L,T\}.
\end{equation}
Setting $E^2\rho_X(E)=\rho^\text{toy}(E)$ allows computation of bounds on $R_X$.
Numerical results will be presented below for $R_L$ only with $m_\tau=0.35$.

For analyticity-based interpolation, the kernel is a rational function of the energy variable on each side of the threshold, so the piecewise-PMP treatment of \cref{sec:thresholds} can be applied directly.
In the moment-problem formulation, a change of variables $\lambda=e^{-E}$ converts the spectral representation of the correlator to
\begin{equation}
\label{eq:CX-tau-moment-def}
    C_X(t)
    =
    \int_0^1 d\lambda\,\tilde{\rho}_X(\lambda)\lambda^t,
\end{equation}
and the target observable becomes
\begin{equation}
\label{eq:RX-tau-moment}
    R_X
    =
    \int_0^1 d\lambda\,
    \tilde{K}_X(\lambda)\tilde{\rho}_X(\lambda).
\end{equation}
The transformed kernel $\tilde{K}_X(\lambda)$ is not polynomial in $\lambda$, so the moment-problem calculation uses a rational approximation to the target kernel. 
For the numerical test, a total of $N_t=20$ timeslices are used. 
For simplicity, an order-20 Taylor expansion of $(1 - \lambda) \tilde{K}_X(\lambda)$ around $\lambda=1$ is used to approximate the target kernel.
For the parameters used here, this approximation reproduces the exact toy value of $\mathcal{K}[\rho]$ to better than one part in $10^{11}$.
This approximation error is therefore negligible for the present demonstration.
Note that the approximation order of the kernel may be increased arbitrarily irrespective of $N_t$.

The kernel has an infrared singularity in the energy variable.
For the toy model, this is avoided by imposing a conservative lower edge of support, $m_\text{gap}=0.05$, which is half the true spectral gap in \cref{eq:toy_spectrum}.
For Nevalinna--Pick-based interpolation, the interval is restricted to $E \in [m_\text{gap},\infty)$.
For moment inputs the corresponding interval is $\lambda\in[0,e^{-m_\text{gap}}]$.
This restriction is part of the definition of the the interval $I$ in \cref{eq:pmp-reduce-dual-moment-rational}.

\begin{figure}
    \centering
    \includegraphics[width=\linewidth]{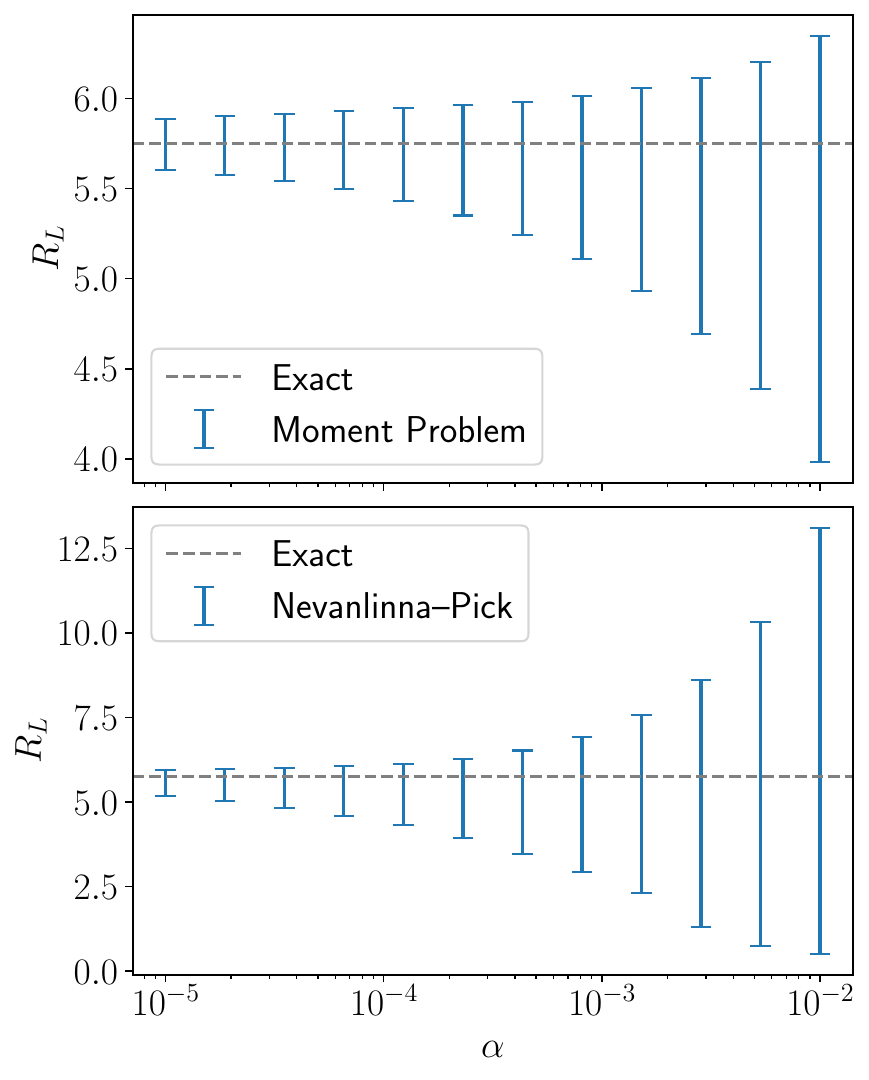}
    \caption{Bounds on the toy inclusive $\tau$-decay observable $R_L$ defined in \cref{eq:RX-tau-defn}.
    The central data are generated from the exact toy spectral density in \cref{eq:rho-toy-defn}, and the input uncertainty is varied through the modeled compatibility region.
    The difference between the size of the bounds is discussed in the main text.
    }
    \label{fig:tau-decay}
\end{figure}

\Cref{fig:tau-decay} shows the resulting bounds on $R_L$ from solving the matrix ($r=2$) versions of \cref{eq:pmp-reduce-dual-moment-rational} and \cref{eq:pmp-reduce-NP-problem} using moment-problem and NP-based inputs, respectively.
In all cases, the bounds contain the exact value computed directly from the toy spectral density. 
The bounds shrink as the input uncertainty is reduced, as expected from the decreasing size of the feasible set.
For this particular observable and finite input data, the moment formulation happens to give substantially tighter bounds than the NP-based formulation.
This should not be interpreted as an intrinsic hierarchy between the methods.
Rather, the two calculations impose different finite-data constraints on the same underlying toy density. 

\subsection{Toy model for hadronic vacuum polarization}\label{sec:HVP-numerics}

The next numerical demonstration models the calculation of the leading-order hadronic vacuum polarization contribution to the anomalous magnetic moment of the muon using the toy model.
This observable can be written in spectral form as
\begin{equation}
\label{eq:hvp-target}
    a_\mu^{\text{HVP}}
    =
    \left(\frac{\alpha_{\text{EM}}}{\pi}\right)^2
    \int_0^\infty ds\, K(s)\rho(s),
\end{equation}
where the kernel $K(s)$ is known analytically in terms of the squared energy $s$ and is reviewed in \cref{sec:hvp-explicit} for completeness.
The corresponding Euclidean correlator admits the representation
\begin{equation}
\label{eq:corr-R-ratio-spectral}
    C(t)
    =
    \frac{1}{2}
    \int_0^\infty ds\,\sqrt{s}\rho(s)e^{-\sqrt{s}t}.
\end{equation}
Changing variables to $\lambda=e^{-\sqrt{s}}$ converts this into a Hausdorff moment problem,
\begin{equation}
\label{eq:hvp-hausdorff-moment}
    C(t)
    =
    \int_0^1 d\lambda\,\tilde{\rho}(\lambda)\lambda^t,
\end{equation}
with
\begin{equation}
    \tilde{\rho}(\lambda)
    =
    \frac{1}{\lambda}(-\log\lambda)^2
    \rho\!\left((-\log\lambda)^2\right).
\end{equation}
In terms of this transformed density,
\begin{equation}
    a_\mu^{\text{HVP}}
    =
    \int_0^1 d\lambda\,
    \tilde{\rho}(\lambda)\tilde{K}(\lambda),
\end{equation}
where
\begin{equation}
    \tilde{K}(\lambda)
    =
    \frac{2}{-\log\lambda}
    K\!\left((-\log\lambda)^2\right).
\end{equation}

As in the preceding example, the transformed kernel is approximated before applying the PMP reduction.
For simplicity, the numerical test uses a 20th-order Chebyshev approximation as implemented in \textsc{mpmath}~\cite{mpmath}.
As above, the approximation error on the kernel was verified to be negligible and can be reduced arbitrarily.
The HVP kernel is also singular at threshold, and the upper-bound calculation is stabilized by imposing the same explicit mass gap $m_\text{gap}=0.05$.
The purpose of the example is not to optimize the treatment of the physical HVP kernel, but to test how the bounds behave for a phenomenologically motivated kernel with nontrivial infrared structure.

\begin{figure}
    \centering
    \includegraphics[width=\linewidth]{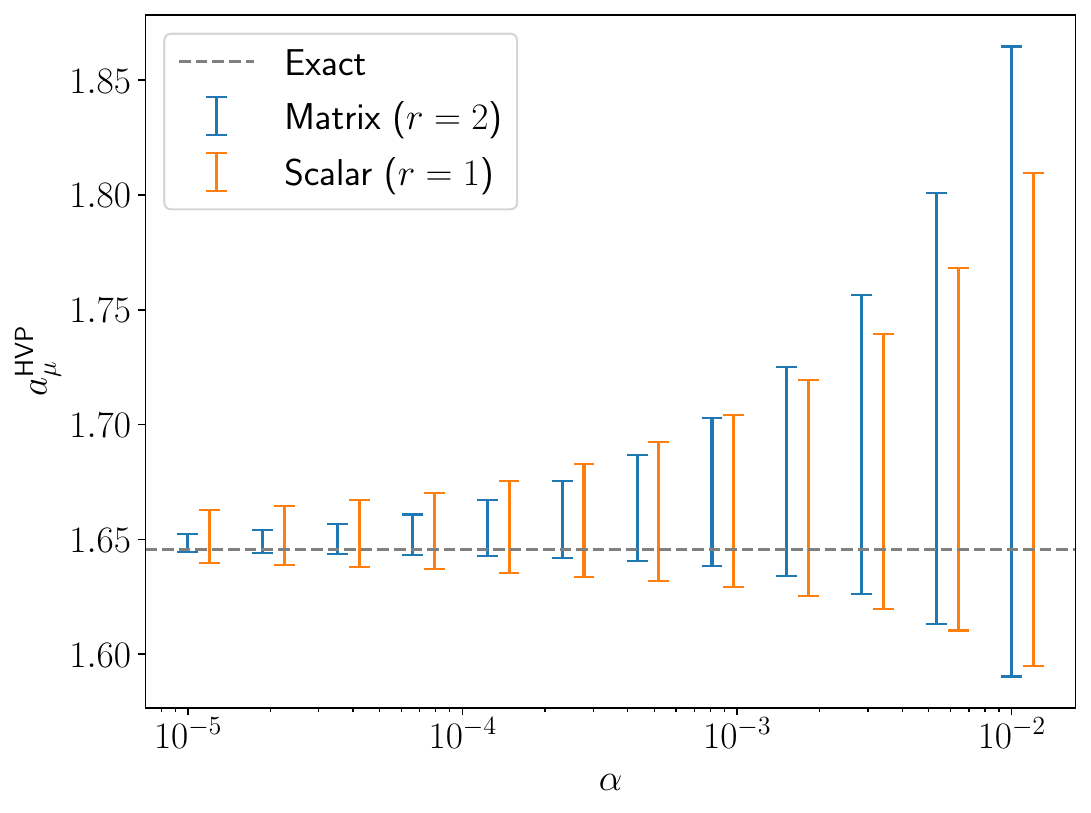}
    \caption{Bounds on the toy HVP-like observable defined in \cref{eq:hvp-target}.
    The scalar bounds use only the $C_{00}(t)$ component of the toy correlator, while the matrix bounds use the full $2\times2$ correlator matrix for the same target spectral component.
    }
    \label{fig:hvp-bounds}
\end{figure}

\Cref{fig:hvp-bounds} shows the bounds resulting from solutions to the scalar (r=1) and matrix $(r=2)$ versions of \cref{eq:pmp-reduce-dual-moment-rational} using input data on $N_t=20$ timeslices.
The exact value computed using the toy spectral function lies within the computed bounds throughout the noise range considered.
As above, the cutoff $\sigma_0$ is scaled for the the scalar and matrix cases according to \cref{eq:cutoff_scaling}.
For larger input uncertainties (roughly the right-hand side of \cref{fig:hvp-bounds}), the scalar and matrix problems comparable bounds.
As the input uncertainty is reduced, the bounds from the matrix problem become stronger, in agreement with expectations.

This behavior illustrates an important qualitative point.
Adding operators is generally expected to improved spectral bounds, since the matrix positivity constraint contains additional information that is not present in a scalar correlator.
On the other hand, a detailed statistical comparison necessarily depends both on the statistical properties of the inputs (via the covariance matrix) and the choice of statistical cutoff (via $\sigma_0$).
The choice of ``equivalent" $\sigma_0$ values for different problem sizes must therefore be made with care.
In this sense, matrix-valued correlators should be viewed as a promising source of additional constraints, not as an automatic improvement independent of statistics and the overall computational cost of generating the matrix of correlators.

%% file: sec6_comparison.tex
\section{Comparison to other methods}\label{sec:comparison}

This section compares the spectral-bootstrap construction with related reconstruction, analyticity, and optimization methods.

\subsection{Kernel reconstruction methods}\label{sec:relation-BG-HLT}

The dual formulation provides an interesting connection to methods that approximate the target smearing kernel as a linear combination of basis functions in order to infer $\mathcal{K}[\rho]$ from the linear combination of the same expansion coefficients with $\hat{C}_t$~\cite{Barata:1990rn,Hansen:2019idp,Gambino:2020crt}.
For simplicity, consider a scalar spectral reconstruction problem with kernel $K:I\to\R$ and basis functions $b_t:I\to\R$.
From a given approximation of the target smearing kernel
\begin{equation}
\label{eq:backus-gilbert-kernel-eq-g}
    \overline{K}(x)
    =
    \sum_{t=1}^{N_t} g_t b_t(x),
\end{equation}
the spectral observable smeared with the approximate $\overline{K}(x)$ is obtained directly from the correlator data:
\begin{equation}
\label{eq:backus-gilbert-kernel-reconstruction}
    \overline{\mathcal{K}}[\rho] =
    \int_I dx\, \rho(x) \overline{K}(x)
    =
    \sum_{t=1}^{N_t}  g_t C_t.
\end{equation}
A common approach is to choose coefficients $g_t$ that approximate the target kernel, while controlling the noise amplification on $\overline{\mathcal{K}}[\rho]$.\footnote{
Note that Refs.~\cite{Bruno:2023bue,Bruno:2024fqc,lupo2026extractionspectraldensitieslattice} decouple the choice of such regulator from the covariance matrix.
Different choices of regulator can be mapped to a corresponding measure of the distance of $C[\rho]$ from $\hat{C}$ in the definition of primal problems and might be worth studying in the future.
}
Often, a covariance penalty is introduced to tame large fluctuations in the coefficients $g_t$.
We note that the noisy-data dual in \cref{eq:opt-problem-inexact-cone-dual} contains a similar term. 
While this term might be interpreted as yielding a similar effect, it bears emphasizing that the present bootstrap formulation chooses $\sigma_0$ as part of the definition of the statistical consistency region;
$\sigma_0$ is not tuned as part of a bias-variance tradeoff.

The dual bootstrap problem provides a similar kernel-reconstruction interpretation, but with two major advantages. 
The first is the fact that, should an approximation of the kernel be required for PMP reduction, the order of the approximation is conceptually unrelated to the number of available input data. 
This approximated kernel will be denoted $\underline{K}$ to distinguish it from $\overline{K}$ above.
The same holds for approximating rational functions. 
We emphasize that the need of quantifying the effects of such approximations, $\overline K$ and $\underline{K}$, is common to both approaches. 

The second advantage is related to the actual role of $N_t$, which instead controls the approximation order of the \emph{bounding} polynomial, cf. \cref{eq:pmp-reduce-dual-moment}.
The upshot is that \emph{exact} inequalities replaces the approximate equality above.

For a lower bound, the pointwise dual feasibility condition is
\begin{equation}
\label{eq:backus-gilbert-kernel-gt-g}
    \underline{K}(x)
    \geq
    \sum_{t=1}^{N_t}  g_t b_t(x)
    \qquad \forall x\in I.
\end{equation}
If $\rho(x)\geq0$, this implies
\begin{equation}
\label{eq:backus-gilbert-kernel-inequality}
    \underline{\mathcal{K}}[\rho] = \int_I dx \, \rho(x) \underline{K}(x)
    \geq
    \sum_{t=1}^{N_t}  g_t C_t.
\end{equation}
Thus a feasible set of dual coefficients gives a one-sided reconstruction of the target observable.
Upper bounds are obtained analogously.
Clearly, the strength of the resulting rigorous bounds on $\mathcal{K}[\rho]$ necessarily scales with the quality of the bounding polynomials, and therefore with $N_t$.

The lattice bootstrap program presented in this work also offers a third advantage over current linear formulations of spectral reconstructions: it naturally extends to matrix-valued spectral functions.
While state-of-the-art linear methods lack this extension, this work might serve as inspiration for future developments.
Building on the analogies observed here between kernel reconstruction methods and the dual programs of this work,
a matrix-valued formulation of the linear methods appears natural.
For instance, one could envision defining matrix-valued approximate kernels, as in \cref{eq:backus-gilbert-kernel-eq-g}, with matrix-valued coefficients $g$.
When applied to the methods of Ref.~\cite{Hansen:2019idp}, for instance, these coefficients could then be obtained from the minimization of a properly modified functional that incorporates the Hermitian trace pairing of the approximate and the target smearing kernels. 

A notable class of smeared observables is constituted by those with an associated kernel whose linear expansion on the Euclidean time basis is exact.
Within this class, time-momentum representations and window observables have played an important role, e.g., in determinations of hadronic contributions to the anomalous magnetic moment of the muon~\cite{Bernecker:2011gh}.
In the moment-problem language, a finite window observable can be defined in terms of a kernel of the form
\begin{equation}
\label{eq:window-kernel-def-comparison}
    K_{\rm win}(\lambda)
    =
    \sum_{t=t_{\min}}^{t_{\max}} w_t\lambda^t
\end{equation}
or as the correspondingly discretized integral.
If the correlator data are known exactly for all times appearing in the sum, then
\begin{equation}
\label{eq:window-obs-defn-comparison}
    \int_0^1 d\lambda\,K_{\rm win}(\lambda)\tilde{\rho}(\lambda)
    =
    \sum_{t=t_{\min}}^{t_{\max}} w_t C_t
\end{equation}
is obtained directly.
In such cases, the dual certificate is exact with $g_t=w_t$, and there is no inverse problem to solve.
However, Euclidean data is not always precisely available for all values of $t$, or the exact coefficients $w_t$ might still yield unstable reconstructions of the smeared observable. 
This could be the case, e.g., for observables that require long-distance correlator information or for kernels with sharp spectral features.
For these cases, casting the problem as a spectral reconstruction may possibly yield benefits, especially if matrix-valued positivity can be leveraged advantageously.

\subsection{Comparison to the conformal bootstrap}

The finite-dimensional reduction used in the present work is closely related to the numerical conformal bootstrap.
Both approaches start from a positivity condition, formulate a polynomial matrix program, reduce it to a finite SDP, and solve the resulting problem with primal-dual interior point methods~\cite{Simmons-Duffin:2015qma,Landry:2019qug}.
In this technical sense, the computational pipeline is the same.

The physical role of the input data is different.
In the conformal bootstrap, positivity and crossing symmetry constrain possible conformal field theories without Monte Carlo data.
In the present work, positivity is combined with finite Euclidean correlator data computed numerically, with no connection to conformal field theory.
In distinction to the conformal bootstrap and to the older $S$-matrix bootstrap program, the present method does not attempt to solve QCD from general principles alone.
Rather, convex optimization is used to propagate the results of nonperturbative calculations into rigorous bounds on smeared spectral observables.

This distinction also explains why two-point functions play a central role here, as opposed the four-point functions typically used in the conformal bootstrap.
In lattice QCD, two-point functions of chosen operators are nontrivial computable data.
The positivity of their spectral representation already gives useful constraints on smeared spectral quantities.
Higher-point functions could be incorporated in future generalizations, but they are not needed for the basic spectral-bounding problem studied in this work.

\subsection{Alternative optimization algorithms}

The dual problems can also be optimized without converting them to PMPs.
For example, Ref.~\cite{Lawrence:2024hjm} considered scalar problems using interior-point or barrier methods applied directly to the dual.
In that approach, the constrained optimization is replaced by an unconstrained or partially constrained optimization with a barrier term that discourages violation of dual feasibility.

Barrier and exchange methods have an important practical advantage.
They do not require the polynomial or rational structure needed for the finite-dimensional PMP reduction.
They can therefore be applied to a broader class of kernels and basis functions.
They may also involve fewer variables than the finite SDP obtained from a sum-of-squares representation.

The tradeoff is certification.
If the positivity condition is enforced by sampling or discretizing the spectral domain, additional work is needed to ensure that no violation occurs between sample points.
Likewise, if the method optimizes only the dual problem, the usual finite-SDP duality gap is not automatically available as a stopping certificate.
These issues are not necessarily fatal, but they require separate control, for instance through adaptive refinement, interval arithmetic, or recovery of a matching primal feasible solution.

It is therefore not obvious a priori which numerical strategy will be best in practical lattice applications.
The PMP-to-SDP route gives strong certificates for polynomial and rational problems, while barrier or exchange methods may be more flexible for exploratory work or for kernels not well suited to polynomial representation.

\subsection{Other approaches to spectroscopy}

Many lattice methods attack spectral problems by directly estimating finite-volume energies and overlap factors.
Examples include generalized eigenvalue methods, multistate fits, and Lanczos/Rayleigh--Ritz methods, or Prony methods.
These methods are closely related to moment problems and to the Krylov-space interpretation of Euclidean correlator data~\cite{Fischer:2020bgv,Wagman:2024rid,Hackett:2024nbe,Ostmeyer:2024qgu,Abbott:2025yhm}.

The goal of the present work is different.
Here the primary output is not a direct calculation of the finite-volume energy spectrum, but instead a rigorous allowed interval for a smeared spectral observable.
As discussed in \cref{sec:recov-prim-solut}, complementary slackness can identify spectra of extremal feasible measures, but those spectra should be interpreted as optimizers of the bounding problem rather than as estimates of the true physical spectrum.

The two perspectives are complementary.
Spectrum-extraction methods can provide detailed physical interpretation and may give very efficient determinations when a small number of states dominate.
Bootstrap bounds avoid committing to a particular finite-state fit ansatz and instead ask what can be concluded from positivity and the available correlator data alone.
Which approach is more useful depends on the observable, the quality of the data, and the extent to which the relevant spectral region is resolved.

%% file: sec7_conclusion.tex
\section{Discussion and outlook}\label{sec:discussion}\label{sec:conclusion}

This work has developed a spectral bootstrap framework for bounding smeared spectral observables from Euclidean lattice data.
Starting from spectral positivity, the method formulates the problem as a convex optimization over all positive spectral densities compatible with the input data.
The dual variables provide checkable one-sided kernel bounds, and polynomial, rational, or piecewise rational instances can be reduced to finite-dimensional semidefinite programs.

The remainder of this Section summarizes the practical assumptions behind the construction, the uncertainties that are included, the additional systematics that remain.

\paragraph*{Interpretation of bounds for noisy data.}
Bounds for noisy data were discussed in \cref{sec:inexact-data}.
The main questions regarding interpretation are how the covariance matrix is estimated, whether it is regularized, and how the cutoff $\sigma_0$ is assigned.
The interpretation also depends on the number of data components.
Care must be taken, for example, when comparing scalar versus matrix bounds.

\paragraph*{Additional uncertainty from kernel approximation.}
A separate source of uncertainty arises when the target kernel is approximated to obtain a finite-dimensional SDP.
As discussed in \cref{sec:finite-diml-reduce}, the finite reduction is exact for polynomial, rational, and piecewise rational kernels satisfying the stated denominator conditions.
For an approximated kernel, the bootstrap bounds apply to the approximating observable; the difference from the desired observable is an additional systematic that must be bounded, monitored by convergence tests, or otherwise included in the final uncertainty budget.

\paragraph*{Reflection positivity.}
The basic physical input used throughout this work is spectral positivity.
For lattice-regulated theories, this positivity follows from reflection positivity of the lattice action and the existence of a positive transfer matrix.
When these conditions hold, the Euclidean correlator admits a positive spectral representation of the form used in \cref{sec:spectral-reconstruction}.

Many large-scale lattice calculations use improved actions that break reflection positivity at finite lattice spacing in order to improve the approach to the continuum limit.
For such actions, the transfer matrix may fail to be Hermitian at finite lattice spacing, and the earliest Euclidean time slices can contain unphysical oscillatory contributions.
Applying the present bounds as exact finite-lattice-spacing statements then requires additional care.
Possible strategies include omitting affected early-time data, inflating their uncertainties, or using a transfer matrix that advances by multiple lattice spacings when such a positive construction is available.

This issue is not unique to the present method.
Many lattice analyses implicitly rely on reflection positivity even when it is not striclty present at finite lattice spacing.
The appropriate level of rigor will depend on the intended application and the size of violations compared to other systematic uncertainties.

\paragraph*{Matrix-valued correlators.}
Matrix correlators are particularly natural in this framework.
Their positive-semidefinite spectral representation supplies constraints that are invisible in a single correlator.
This is different from purely linear reconstructions, where additional operators do not affect the target observable unless they enter the chosen estimator.
The toy HVP example in \cref{sec:HVP-numerics} illustrates that matrix information can tighten the bounds substantially in favorable regimes.

Such improvement is not necessarily automatic.
A larger correlator matrix also has more covariance components, may require a larger statistical cutoff at fixed confidence level, and may be more expensive to compute.
Whether additional operators are worthwhile (e.g., compared to increasing statistics or extending the temporal dimension for fixed correlator basis)
depends on their overlaps with the spectral region relevant to the target kernel, their noise properties, and the cost of computing the full matrix.
Understanding this tradeoff is an important practical question for future applications.

%% file: app_finite-dim.tex
\section{Finite-dimensional SDP reduction}\label{app:finite-dim}

This appendix collects the technical details used in \cref{sec:finite-diml-reduce} to reduce polynomial matrix programs to finite-dimensional semidefinite programs.
The construction follows the methods used in conformal-bootstrap applications of polynomial matrix programs~\cite{Simmons-Duffin:2015qma}, with minor modifications for the intervals and noisy-data constraints used here.

Following Ref.~\cite{Simmons-Duffin:2015qma}, these results are stated for symmetric matrix polynomials $M(x)\in \symmat{r}[x]$, but as proven in Ref.~\cite{DETTE2002169} the same results hold for Hermitian matrix polynomials $M \in \hermmat{r}[x]$, provided that the certificate matrices $Y_a$ are taken to be Hermitian rather than symmetric.

Let $q_0,q_1,\dots$ be a basis of real polynomials with $\deg q_i=i$.
Define the matrix of polynomial products
\begin{equation}
    Q_{ij}(x)=q_i(x)q_j(x).
\end{equation}
The monomial basis $q_i(x)=x^i$ is the simplest conceptual choice, although orthogonal polynomial bases are typically preferred in numerical implementations.

For a polynomial matrix $M(x)\in\symmat{r}[x]$, positivity on the real line can be represented by a matrix sum of squares.
One convenient form is
\begin{lemma}
\label{lem:trY-sos-block}
Let $M(x)\in\symmat{r}[x]$ have degree $d$.
Then $M(x)\succeq0$ for all $x\in\R$ if and only if there exists a positive-semidefinite block matrix $Y \in \symmat{\floor{d/2}}$ such that
\begin{equation}
\label{eq:pos-poly-sos-trY}
    M(x)
    =
    \Tr'\!\left[
        Y\left(Q(x)\otimes \mathds{1}_r\right)
    \right],
\end{equation}
where $\mathds{1}_r$ is the $r\times r$ identity matrix and $\Tr'$ denotes the partial trace over the polynomial-basis indices, leaving an $r\times r$ matrix.
\end{lemma}
Equivalently, one may write
\begin{equation}
    M(x)
    =
    V(x)^\top Y V(x),
\end{equation}
where $V(x)=(q_0(x),q_1(x),\dots,q_m(x))\otimes \mathds{1}_r$ for $m = \floor{d/2}$.

For a half-line, the corresponding representation is
\begin{lemma}
\label{lem:trY-sos-block-stieltjes}
Let $M(x)\in\symmat{r}[x]$ have degree $d$.
Then $M(x)\succeq0$ for all $x\geq0$ if and only if there exist positive-semidefinite block matrices $Y_1 \in \symmat{\floor{d/2}+1}$ and $Y_2 \in \symmat{\floor{(d-1)/2}+1}$ such that
\begin{equation}
\label{eq:pos-poly-sos-trY-stieltjes}
    M(x)
    =
    \Tr'\!\left[
        Y_1\left(Q(x)\otimes \mathds{1}_r\right)
    \right]
    +
    x\,
    \Tr'\!\left[
        Y_2\left(Q(x)\otimes \mathds{1}_r\right)
    \right].
\end{equation}
\end{lemma}
By translating the variable, the same representation holds for any half-line $[a,\infty)$ with the lemma applied to $M(a+x)$ with $x\geq0$.

For bounded intervals, one may use the standard Hausdorff-moment positivity certificates.
For the interval $[0,1]$, the result can be written as follows~\cite{DETTE2002169}:
\begin{lemma}[Dette--Studden]
\label{lem:trY-sos-block-hausdorff}
Let $M(x)\in\symmat{r}[x]$ have degree $d$.
If $d$ is even, then $M(x)\succeq0$ for all $x\in[0,1]$ if and only if there exist positive-semidefinite block matrices $Y_1 \in \symmat{d/2+1}$ and $Y_2 \in \symmat{d/2}$ such that
\begin{equation}
\label{eq:pos-poly-sos-trY-hausdorff-even}
    M(x)
    =
    \Tr'\!\left[
        Y_1\left(Q(x)\otimes \mathds{1}_r\right)
    \right]
    +
    x(1-x)
    \Tr'\!\left[
        Y_2\left(Q(x)\otimes \mathds{1}_r\right)
    \right].
\end{equation}
If $d$ is odd, then $M(x)\succeq0$ for all $x\in[0,1]$ if and only if there exist positive-semidefinite block matrices $Y_1 \in \symmat{\floor{d/2}+1}$ and $Y_2 \in \symmat{\floor{d/2}+1}$ such that
\begin{multline}
\label{eq:pos-poly-sos-trY-hausdorff-odd}
    M(x)
    =
    x\,
    \Tr'\!\left[
        Y_1\left(Q(x)\otimes \mathds{1}_r\right)
    \right]
    \\
    +
    (1-x)
    \Tr'\!\left[
        Y_2\left(Q(x)\otimes \mathds{1}_r\right)
    \right].
\end{multline}
\end{lemma}
The corresponding certificate on a general bounded interval $[a,b]$ is obtained by an affine change of variables mapping $[a,b]$ to $[0,1]$.

%% file: app_primal-dual.tex
\section{Primal-dual interior point methods}\label{sec:primal-dual}

This appendix reviews the solution of finite-dimensional SDPs like those in \cref{sec:finite-diml-reduce} using primal-dual interior point methods.
\cref{sec:primal_dual_finite_sdp} begins by defining a dual version of the finite SDP in \cref{eq:dual-pmp-sdp-full} and defining the duality gap.
\cref{sec:interior_point_methods} then describes their solution using primal-dual interior point methods; \cref{sec:interior_point_methods} follows closely the presentation of Ref.~\cite{Simmons-Duffin:2015qma}, with a few modifications specific to the SDPs encountered in this work.

\subsection{Primal-dual finite SDP and the duality gap \label{sec:primal_dual_finite_sdp}}

Consider a generic finite SDP of the form
\begin{equation}
    \label{eq:dual-finite-sdp}
\begin{aligned}
    \maximize_{y, Y} \quad &
    b\cdot y + D \matdot Y \\
    \st \quad & \Tr[A_p Y] + (By)_p = c_p, \\&\forall p \in \{1,\dots, P\},\\
    & Y \succeq 0,
\end{aligned}
\end{equation}
for which the Lagrange dual has the form
\begin{equation}
\label{eq:primal-pmp-sdp-full}
\begin{aligned}
    \minimize_{x} \quad & c\cdot x \\
    \st \quad& X=\sum_p A_p x_p-D,\\
    &B^\top x=b,\\
    &X\succeq0.
\end{aligned}
\end{equation}
This finite-dimensional dual pair should not be confused with the original primal and dual spectral reconstruction problems.
The terminology is inherited from the SDP solver: \cref{eq:primal-pmp-sdp-full} is the primal finite SDP, while \cref{eq:dual-finite-sdp} is the dual finite SDP.

If $(x,X)$ and $(y,Y)$ are feasible for the finite primal and dual SDPs, then
\begin{equation}
\label{eq:primal-dual-gap-SDP}
    c\cdot x - D\matdot Y - b\cdot y
    =
    X\matdot Y
    \geq0.
\end{equation}
The left-hand side is the finite-SDP duality gap.
Primal-dual interior point methods monitor this quantity during optimization.
A sufficiently small duality gap, together with feasibility of the finite SDP variables, certifies that the finite SDP has been solved to the corresponding precision.

Further implementation details of the primal-dual interior point method used in the numerical examples are reviewed in \cref{sec:primal-dual}.

\subsection{Interior-point methods \label{sec:interior_point_methods}}

The starting point for primal-dual interior point methods is the finite SDP dual pair in \cref{eq:dual-finite-sdp,eq:primal-pmp-sdp-full}.
For feasible primal and dual SDP variables, the duality gap is $X\matdot Y\geq0$.
Since $X,Y\succeq0$, this gap vanishes if and only if the matrix product $XY$ vanishes.

Primal-dual methods replace the exact complementarity condition by the deformed equation
\begin{equation}
\label{eq:mu-relaxed-complementarity}
    XY = \mu \mathds{1}_{\dim X},
\end{equation}
where $\mu>0$ and $\mathds{1}_{\dim X}$ is the identity matrix of the same dimension as $X$.
Together with primal and dual feasibility, \cref{eq:mu-relaxed-complementarity} defines the central path
\begin{equation}
    (x(\mu),X(\mu),y(\mu),Y(\mu)).
\end{equation}
The optimization proceeds by Newton steps that push the variables toward the central path while gradually taking $\mu\to0$.

To derive the Newton step, expand
\begin{equation}
    (x,X,y,Y)\mapsto(x+dx,X+dX,y+dy,Y+dY)
\end{equation}
and solve the linearized central-path equations.
It is useful to define the residuals
\begin{align}
    P &= \sum_p A_p x_p - X - D, \\
    p &= b - B^\top x, \\
    d_i &= c_i - \Tr[A_iY] - (By)_i, \\
    R &= \mu I - XY.
\end{align}
In terms of these residuals, the vector steps $dx$ and $dy$ are obtained from
\begin{equation}
\label{eq:schur-compl-solve-dxdy}
    \begin{pmatrix}
        S & -B \\
        B^\top & 0
    \end{pmatrix}
    \begin{pmatrix}
        dx \\
        dy
    \end{pmatrix}
    =
    \begin{pmatrix}
        -d_i-\Tr[A_iZ] \\
        p
    \end{pmatrix},
\end{equation}
where $Z=X^{-1}(PY-R)$, and $S$ is referred to as the Schur complement and defined by
\begin{equation}
\label{eq:schur-complement-S-defn}
    S_{ij} = \Tr[A_iX^{-1}A_jY].
\end{equation}
For the rank-one choices of $A_i$ used in this work, the trace in \cref{eq:schur-complement-S-defn} can be evaluated efficiently by changing the order of index contraction, as in Ref.~\cite{Simmons-Duffin:2015qma}.
Once $dx$ is known, the matrix steps are
\begin{align}
    dX &= P+\sum_i A_i dx_i, \\
    dY &= X^{-1}(R-dX\,Y).
\end{align}
The expression for $dY$ is not necessarily symmetric.
Following SDPB, we symmetrize $dY$ before taking the step.

There are several implementation details in which the solver used here differs from the original implementation of Ref.~\cite{Simmons-Duffin:2015qma}.
First, the implementation allows only a subset of the linear equalities to be enforced, as needed for the noisy-data SDP constraints in \cref{eq:opt-problem-def-inexact-block-dual}.
Second, following later versions of SDPB, we do not add the identity matrix to the left-hand side of \cref{eq:schur-compl-solve-dxdy}~\cite{Landry:2019qug}.
Third, we use the same step size for the primal and dual variables, taking the smaller allowed value when necessary.
Occasional stalling is handled by increasing $\mu$ by a factor of $1.5$ and holding it fixed for ten Newton iterations.
All numerical examples in this work use 150 decimal digits of precision.

The final modification is the elimination of free variables.
This reduction was suggested in Ref.~\cite{Simmons-Duffin:2015qma}, following Ref.~\cite{kobayashi2007conversion}, but was not used there because of the structure of conformal-bootstrap SDPs.
In the problems considered here, eliminating the free variables reduces the problem size and improves numerical stability.

The essential idea is to split the primal vector as
\begin{equation}
    x = x_\parallel+x_\perp
\end{equation}
such that
\begin{equation}
    B^\top x_\parallel=b,
    \qquad
    B^\top x_\perp=0.
\end{equation}
The variables $x_\parallel$ can then be eliminated, which also eliminates the free variables $y$.
Rather than constructing the relevant subspaces from linearly independent columns of $B$, we use an SVD factorization
\begin{equation}
    B = U\diag(\sigma_i)\tilde{V}^\top = UV^\top,
\end{equation}
where the singular values have been absorbed into $V^\top$.
The decomposition is arranged so that $V$ is square and invertible, and
\begin{equation}
    U =
    \begin{pmatrix}
        U_\parallel & U_\perp
    \end{pmatrix},
\end{equation}
with $U_\parallel$ square of the same dimension as $V$.

With this decomposition, the substitutions
\begin{align}
    \tilde{x} &= U_\parallel x, \\
    \tilde{c} &= U_\perp c, \\
    \tilde{D} &= \sum_p A_p U_{\parallel,pq}(V^{-\top}b)_q, \\
    \alpha_0 &= c^\top U_\parallel V^{-\top}b, \\
    \tilde{A}_q &= \sum_p (U_\perp)_{pq}A_p
\end{align}
reduce the original primal finite SDP to
\begin{equation}
\label{eq:primal-SDP-free-reduced}
\begin{aligned}
    \text{minimize }& \alpha_0+\tilde{c}\cdot\tilde{x} \\
    \text{subject to }& X=\sum_q \tilde{A}_q\tilde{x}_q-\tilde{D},\\
    &X\succeq0.
\end{aligned}
\end{equation}
Dualizing gives the reduced dual SDP
\begin{equation}
\label{eq:dual-SDP-free-reduced}
\begin{aligned}
    \text{maximize }& \tilde{D}\matdot Y \\
    \text{subject to }& \Tr[\tilde{A}_pY]=\tilde{c}_p,
    \quad \forall p.
\end{aligned}
\end{equation}
The Schur complement for the reduced problem can be computed from the unreduced Schur complement by
\begin{equation}
    \tilde{S}(X,Y)=U_\perp^\top S(X,Y)U_\perp.
\end{equation}
This projection breaks some of the original block structure of $S$, as noted in Ref.~\cite{Simmons-Duffin:2015qma}.
For the SDPs considered here, the reduced dimension and improved conditioning outweigh this cost.

%% file: app_equivalence_proof.tex
\section{Equivalent descriptions of the Wertevorrat
}\label{app:equivalence_proof}

This appendix contains a proof of \cref{example:moment_bounds} and numerical demonstrations of both \cref{example:moment_bounds,example:interpolation_bounds}.
Since both sides of the equivalence can always be computed explicitly, the proof 
focuses on structural aspects, 
ignoring inessential details including measure-theoretic details of convergence and regularity and existence or relevant minimizations.

\begin{proof}[Proof of \cref{example:moment_bounds}]
Suppose that the moments $\hat C_t$ are given for a scalar non-extremal
truncated Hamburger moment problem with $N=2n$.  Let
$ \mathcal{F}=\{\rho\in\mathcal M^1_+(\R): C_t[\rho]=\hat C_t,\ t=0,\ldots,N\}$
be the set of feasible spectral functions.
For $z$ in the upper half-plane,
define the set of possible values of the Green's function in \cref{eq:stieltjes-trans-def} by
\begin{equation}
\label{eq:Wertevorrate-definition-app}
  \Delta_N(z) = \cl\{G_\rho(z):\rho\in\mathcal{F}\}.
\end{equation}
Kovalishina~\cite{Kovalishina1984} has proven that $\Delta_N(z)$ is a disk
\begin{equation}
\label{eq:Wertevorrate-disk-app}
  \Delta_N(z)=\bar{\D}(c,r)
  =
  \{\zeta\in\C:|\zeta-c|\leq r\},
\end{equation}
with known analytic formulas for the center $c$ and radius $r$ in terms of the moments $\hat{C}_t$.

On the other hand, the SDP reconstruction is obtained from the following intersection of half-planes:
\begin{equation}
  \D'
  =
  \bigcap_{\theta\in[0,2\pi)}
  \left\{
    \zeta\in\C:
    \real(e^{i\theta}\zeta)\geq \ell(\theta)
  \right\}
\end{equation}
where, for each $\theta\in[0,2\pi)$, we define the supporting lower bound
\begin{equation}
  \ell(\theta)
  = \inf_{\rho\in\mathcal{F}} \real\left(e^{i\theta}G_\rho(z)\right).
\end{equation}
Equivalently, assuming the relevant extremal values are attained after taking
the closure and using \cref{eq:Wertevorrate-definition-app} we obtain
\begin{equation}
  \ell(\theta)
  =
  \min_{\zeta\in\Delta_N(z)}
  \real(e^{i\theta}\zeta).
\end{equation}
For each $\theta$, the minimum of
$\real(e^{i\theta}\zeta)$ 
is attained at the boundary point in
the direction opposite to $e^{i\theta}$, and hence
\begin{equation}
  \ell(\theta)
  =
  \min_{\zeta\in\bar{\D}(c,r)}
  \real(e^{i\theta}\zeta)
  =
  \real(e^{i\theta}c)-r.
\end{equation}
Thus the SDP half-plane reconstruction is
\begin{equation}
  \D'
  =
  \bigcap_{\theta\in[0,2\pi)}
  \left\{\zeta\in\C:
    \real(e^{i\theta}(\zeta-c))\geq -r
  \right\}.
\end{equation}
If $\zeta\in\bar{\D}(c,r)$, then
\begin{equation}
  \real(e^{i\theta}(\zeta-c))
  \geq -|\zeta-c|
  \geq -r
  \qquad \forall\theta,
\end{equation}
so $\zeta\in\D'$, i.e. $\D(c,r) \subseteq \D'$.  Conversely, suppose $\zeta\in\D'$.
If $\zeta=c$, then $\zeta\in \bar{\D}(c,r)$ trivially.
Suppose therefore $x=\zeta-c\neq 0$.
Choosing $\theta=\pi-\arg x$ gives
\begin{equation}
  -|x|=\real(e^{i\theta}x)\geq-r.
\end{equation}
Therefore $|\zeta-c|=|x|\leq r$, so $\zeta\in\bar{\D}(c,r)$,
i.e. $\D' \subseteq \D(c,r)$.
Hence
\begin{equation}
  \D'=\bar{\D}(c,r),
\end{equation}
which was to be shown.
\end{proof}

\begin{figure*}[t!]
    \centering
    \includegraphics[width=0.37\linewidth]{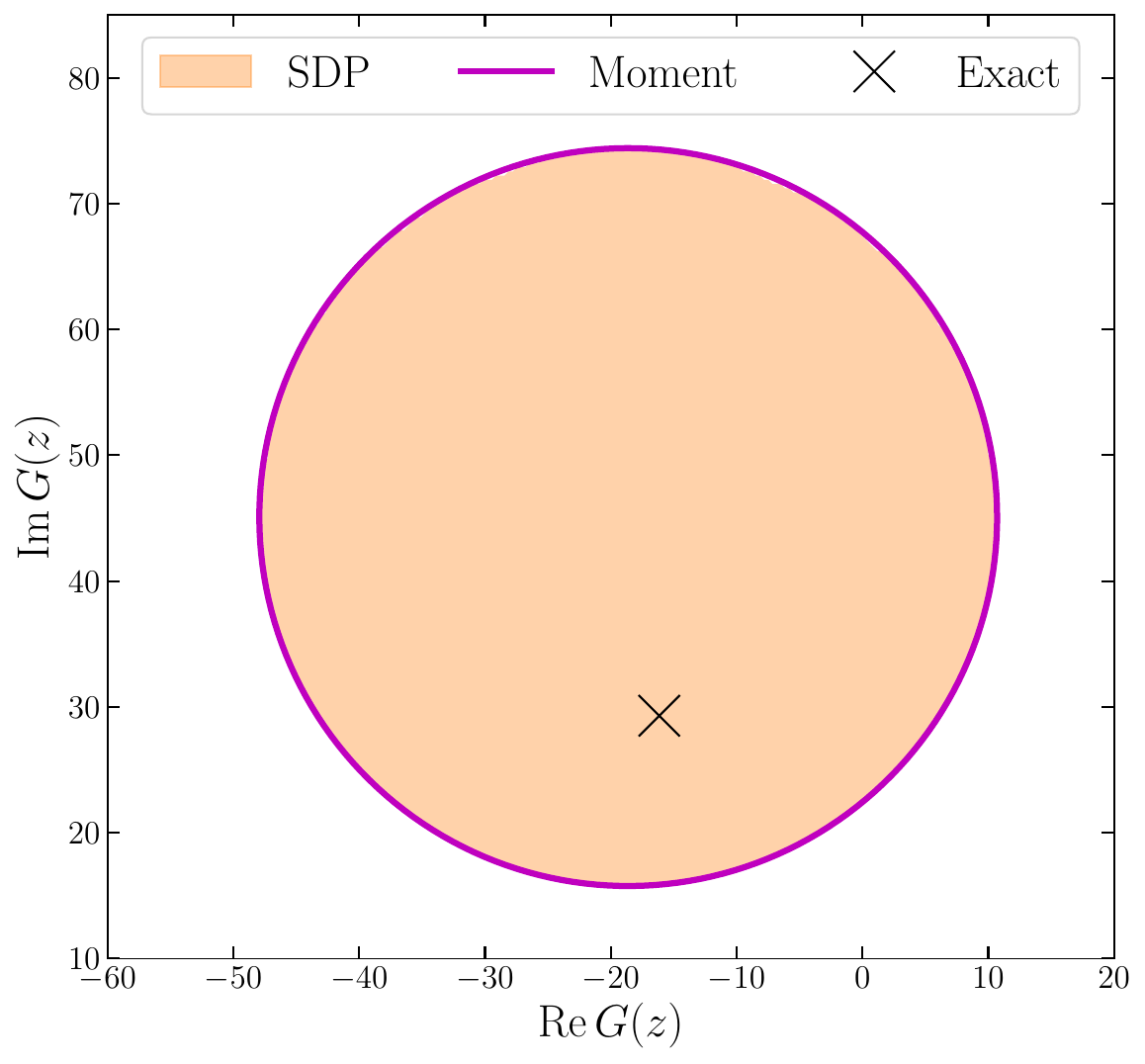}
    \includegraphics[width=0.35\linewidth]{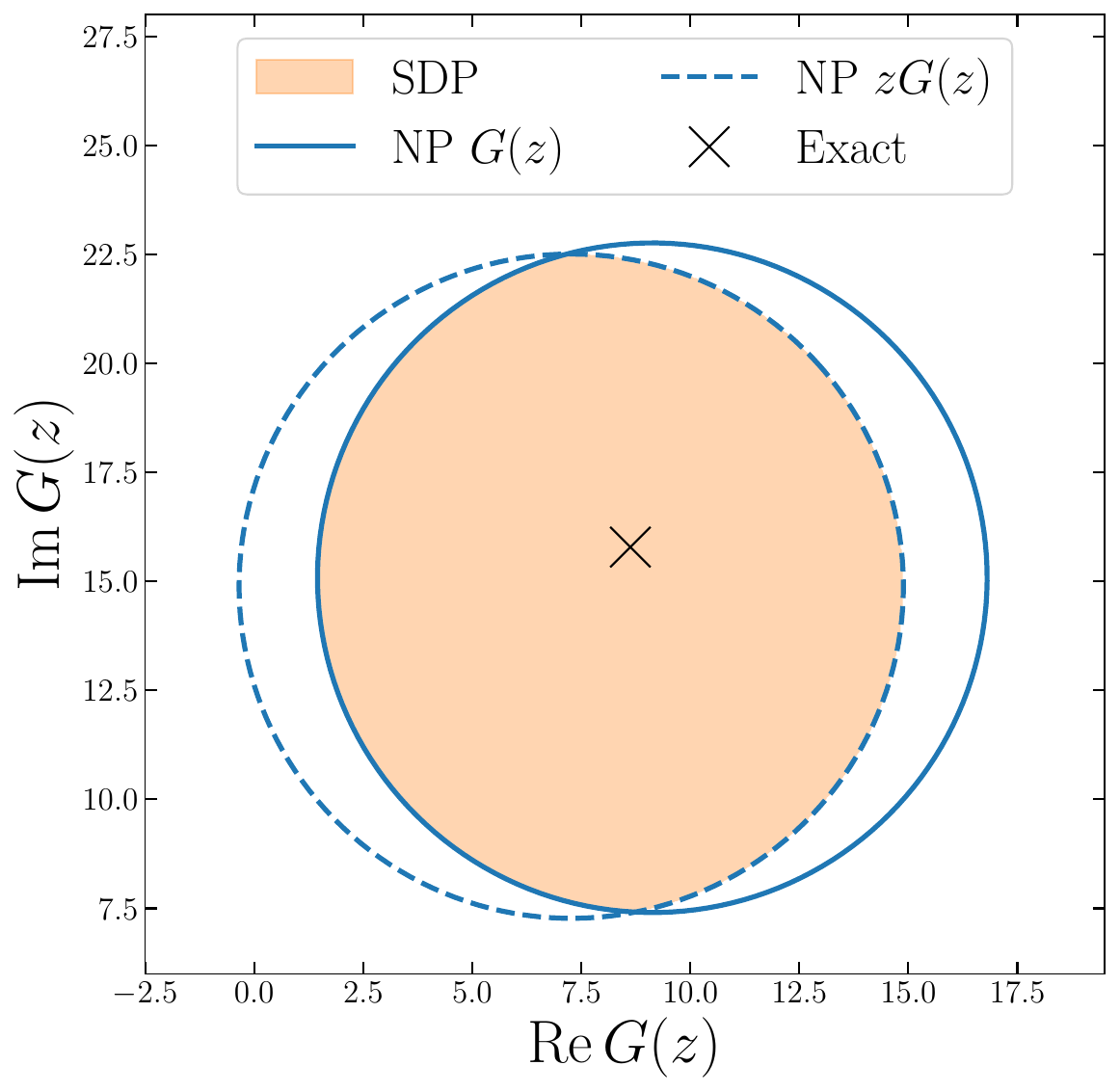}
    \caption{Numerical demonstration of the equivalent formulation of bounds on $G(z)$ in \cref{eq:stieltjes-trans-def}.
    The left panel shows the two different methods in \cref{example:moment_bounds} related to the moment problem.
    The right panel shows the two different methods related to Nevanlinna--Pick interpolation in \cref{example:interpolation_bounds}.
    }
    \label{fig:equivalent_bounds}
\end{figure*}

\Cref{fig:equivalent_bounds} confirms the equivalence of the bounds for $G(z)$ from \cref{eq:stieltjes-trans-def} computed using \cref{example:moment_bounds} and \cref{example:interpolation_bounds} at a fixed locations in the upper half-plane, respectively $z=E+i\epsilon$ and $z=\lambda+i\epsilon$ at $E=0.4$, $\epsilon=0.1$ with $\lambda=e^{-E}$.
For simplicity, the numerical example used the toy model from \cref{sec:toy-model} with $N_{\rm states}=20$ and exact data on $N_t=8$ timeslices.

%% file: app_np-moment.tex
\section{Nevanlinna--Pick data from moment data}\label{sec:NP-moment}

This appendix records a direct relation between finite moment data and and input data for analyticity-based interpolation methods like Nevanlinna--Pick interpolation.
The relation is useful for comparing the two formulations, but it also highlights an important practical point: finite moment data and independently supplied values of a Stieltjes transform are not the same input unless the transformation between them is explicitly performed.

Start from a Hausdorff moment problem on $[0,\infty)$,
\begin{equation}
\label{eq:NP1}
    C_t[\rho]
    =
    \int_0^1 d\lambda\,\rho(\lambda)\lambda^t,
    \qquad
    t=0,\dots,N_t-1.
\end{equation}
Let $\omega_n=2\pi n/N_t$ for bosonic frequencies, or $\omega_n=2\pi(n+\tfrac12)/N_t$ for fermionic frequencies.
Define the discrete Fourier transform of the moment data by
\begin{equation}
\label{eq:fourier-trans-defn}
    C_{\omega_n}[\rho]
    =
    \sum_{t=0}^{N_t-1}
    e^{-i\omega_n t} C_t[\rho].
\end{equation}
Substituting \cref{eq:NP1} into \cref{eq:fourier-trans-defn} gives
\begin{align}
    C_{\omega_n}[\rho]
    &=
    \int_0^1 d\lambda\,\rho(\lambda)
    \sum_{t=0}^{N_t-1}
    \left(\lambda e^{-i\omega_n}\right)^t
    \nonumber \\
    &=
    \int_0^1 d\lambda\,\rho(\lambda)
    \frac{1-\left(\lambda e^{-i\omega_n}\right)^{N_t}}
    {1-\lambda e^{-i\omega_n}}
    \nonumber \\
    &=
    e^{i\omega_n}
    \int_0^1 d\lambda\,\rho(\lambda)
    \frac{1\mp\lambda^{N_t}}
    {e^{i\omega_n}-\lambda}.
    \label{eq:C-omega-l-geom-sum-result}
\end{align}
The upper or lower sign is determined by whether $e^{i\omega_n N_t}=+1$ or $-1$.

Define
\begin{equation}\label{eq:NP_zn}
    z_n=e^{i\omega_n}
\end{equation}
and introduce the modified non-negative density
\begin{equation}
\label{eq:rhotilde-def-NP-moment}
    \tilde{\rho}(\lambda)
    =
    (1\mp\lambda^{N_t})\rho(\lambda),
\end{equation}
with the same sign convention as in \cref{eq:C-omega-l-geom-sum-result}.
Then
\begin{equation}
\label{eq:C-omega-l-eq-stieltjes}
    C_{\omega_n}[\rho]
    =
    z_n
    \int_0^1 d\lambda\,
    \frac{\tilde{\rho}(\lambda)}
    {z_n-\lambda}.
\end{equation}
Up to the overall factor of $z_n$, this is the Stieltjes transform of $\tilde{\rho}$ evaluated at $z_n$.
With the convention
\begin{equation}
\label{eq:NP-moment-stieltjes}
    \tilde{G}(z)
    =
    \int_0^1 d\lambda\,
    \frac{\tilde{\rho}(\lambda)}
    {z-\lambda},
\end{equation}
the Fourier-transformed moment data specify
\begin{equation}
    C_{\omega_n}[\rho]=z_n \tilde{G}(z_n).
\end{equation}

There is redundancy in these interpolation conditions.
For real moment data, $\tilde{G}(z^*)=\tilde{G}(z)^*$ and $z_n^*=z_{N_t-n}$.
A non-redundant set of conditions can therefore be taken from the nodes in the upper half of the unit circle.
In this sense, finite moment data can be repackaged as interpolation data for the Stieltjes transform on the unit circle.

The present formulation of the Nevanlinna-Pick interpolation problem on $\tilde{G}$ is directly accessible from Euclidean time correlator data. 
In contrast to its continuum counterpart in the energy variable $E$ (cf. Section~VI in Ref~\cite{Bergamaschi:2023xzx}),
the interpolation is phrased exactly in the upper half-plane of the complex variable $\lambda$.

This construction explains why the two approaches can be compared within the same convex-optimization framework.
It also explains why bounds obtained from independently chosen moment and Nevanlinna--Pick data need not coincide at finite data.
They agree only when the two sets of inputs encode the same linear information about the same positive density and the same target functional. 

%% file: app_hvp-kernel.tex
\section{HVP kernel}\label{sec:hvp-explicit}

This appendix records the kernel used for the hadronic-vacuum-polarization example in \cref{sec:HVP-numerics}.
The leading-order hadronic vacuum polarization contribution to the anomalous magnetic moment of the muon can be written as~\cite{Blum:2002ii,Aliberti:2025beg,Bernecker:2011gh}
\begin{equation}
\label{eq:hvp-LO-Pi}
    a_\mu^{\text{HVP}}
    =
    \left(\frac{\alpha}{\pi}\right)^2
    \int_0^\infty dQ^2\,K_E(Q^2)\hat{\Pi}(Q^2),
\end{equation}
where
\begin{align}
    K_E(s)
    &=
    \frac{1}{m_\mu^2}
    \hat{s}Z(\hat{s})^3
    \frac{1-\hat{s}Z(\hat{s})}{1+\hat{s}Z(\hat{s})^2},
    \\
    Z(\hat{s})
    &=
    \frac{\hat{s}-\sqrt{\hat{s}^2+4\hat{s}}}{2\hat{s}},
    \qquad
    \hat{s}=\frac{s}{m_\mu^2}.
\end{align}
Here $m_\mu$ is the muon mass and
\begin{equation}
    \hat{\Pi}(Q^2)
    =
    4\pi^2\left[\Pi(Q^2)-\Pi(0)\right]
\end{equation}
is the subtracted vacuum polarization, restricted to the hadronic contribution.

The electromagnetic vector-current correlator admits the spectral representation
\begin{equation}
\label{eq:corr-R-ratio-spectral-appendix}
    C(t)
    =
    \frac{1}{2}
    \int_0^\infty ds\,\sqrt{s}\rho(s)e^{-\sqrt{s}t},
\end{equation}
where the positive spectral density $\rho(s)$ is proportional to the usual $R$-ratio spectral function.
The same spectral density enters the dispersion relation
\begin{equation}
\label{eq:Pi-dispersion}
    \Pi(Q^2)-\Pi(0)
    =
    Q^2
    \int_0^\infty
    \frac{ds\,\rho(s)}
    {s(s+Q^2)}.
\end{equation}
Substituting \cref{eq:Pi-dispersion} into \cref{eq:hvp-LO-Pi} gives the spectral representation
\begin{equation}
    a_\mu^{\text{HVP}}
    =
    \left(\frac{\alpha}{\pi}\right)^2
    \int_0^\infty ds\,K(s)\rho(s),
\end{equation}
with
\begin{equation}
\label{eq:HVP-kernel-s}
    K(s)
    =
    \frac{4\pi^2}{s}
    \int_0^\infty dQ^2\,
    \frac{Q^2K_E(Q^2)}
    {s+Q^2}.
\end{equation}

Let
\begin{equation}
    \tau=\frac{s}{4m_\mu^2},
\end{equation}
and
\begin{align}
    x
    &=
    \frac{1-\beta_\mu}{1+\beta_\mu},
    \\
    \beta_\mu
    &=
    \sqrt{1-\frac{4m_\mu^2}{s}}.
\end{align}
The numerical examples in \cref{sec:HVP-numerics} take $m_\mu=0.2$.
The kernel can be evaluated in closed form in terms of these variables, taking the form~\cite{Aliberti:2025beg,Brodsky:1967sr,Lautrup:1968tdb,Gourdin:1969dm,Bouchiat:1961lbg}:
\begin{widetext}
\begin{equation}
\label{eq:HVP-kernel-closed-form}
    K(s)
    =
    \begin{cases}
    \frac{1}{2}
    -4\tau
    -4\tau(1-2\tau)\log(4\tau)
    -2(1-8\tau+8\tau^2)
    \sqrt{\frac{\tau}{1-\tau}}
    \cos^{-1}(\sqrt{\tau}),
    & 0\leq s\leq4m_\mu^2,\\
    \frac{x^2}{2}(2-x^2)
    +
    \frac{(1+x^2)(1+x)^2}{x^2}
    \left[
        \log(1+x)-x+\frac{x^2}{2}
    \right]
    +
    \frac{1+x}{1-x}x^2\log x,
    & s\geq4m_\mu^2.
    \end{cases}
\end{equation}
\end{widetext}